\documentclass[10pt]{amsart}

\usepackage{amssymb, latexsym, mathtools}
\usepackage{epsfig}
\usepackage{paralist} 
\usepackage{stmaryrd}
 \usepackage{multicol}
\usepackage{geometry}
 \makeatletter
\newtheorem*{rep@theorem}{\rep@title}
\newcommand{\newreptheorem}[2]{%
\newenvironment{rep#1}[1]{%
 \def\rep@title{#2 \ref{##1}}%
 \begin{rep@theorem}}%
 {\end{rep@theorem}}}
\makeatother
 \usepackage{wrapfig}
\usepackage{pdfsync}
\usepackage{parskip}
\usepackage{txfonts}
\usepackage[T1]{fontenc}
\usepackage{booktabs}
 \usepackage{xcolor}
 \usepackage{tabularx, colortbl}
 \usepackage{graphicx}
 \usepackage[scriptsize, up]{caption}
\usepackage{pgf,tikz}
\usetikzlibrary{decorations, decorations.pathmorphing}
\usetikzlibrary {shapes}

\definecolor{dnrbl}{rgb}{0,0,0.3}
\definecolor{dnrgr}{rgb}{0,0.3,0}
\definecolor{dnrre}{rgb}{0.5,0,0}
\usepackage[colorlinks=true, citecolor=dnrgr, linkcolor=dnrre, urlcolor=dnrre] {hyperref}

\theoremstyle{plain}
\newtheorem{thm}{Theorem}[section]

\newtheorem{lem}[thm]{Lemma}
\newtheorem{coro}[thm]{Corollary}

\numberwithin{equation}{subsection}
\makeatletter

\let\c@table\c@figure
\makeatother

\newreptheorem{lemma}{Lemma}

\newtheorem{defin}[thm]{Definition}

\def\Pu{P_{\textrm{\footnotesize{unhap}}}}
\def\Ps{P_{\textrm{\footnotesize{stab}}}}
\newcommand{\bpm}{\begin{pmatrix}}
\newcommand{\epm}{\end{pmatrix}}

\newcommand{\eps}{\varepsilon}

\def\Pu{P_{\textrm{\footnotesize{unhap}}}}
\def\Ps{P_{\textrm{\footnotesize{stab}}}}

\newcommand{\ds}{\displaystyle}

\renewenvironment{abstract}
 { \normalsize
  \list{}{
    \setlength{\leftmargin}{.0cm}%
    \setlength{\rightmargin}{\leftmargin}%
    }%
  \item {\scshape \abstractname.} \relax}
 {\endlist}
 
\makeatletter
\renewcommand{\@makecaption}[2]{%
  \setbox\@tempboxa\vbox{\color@setgroup
    \advance\hsize-2\captionindent\noindent
    \@captionfont\@captionheadfont#1\@xp\@ifnotempty\@xp
        {\@cdr#2\@nil}{.\@captionfont\upshape\enspace#2}%
    \unskip\kern-2\captionindent\par
    \global\setbox\@ne\lastbox\color@endgroup}%
  \ifhbox\@ne 
    \setbox\@ne\hbox{\unhbox\@ne\unskip\unskip\unpenalty\unkern}%
  \fi
  \ifdim\wd\@tempboxa=\z@ 
    \setbox\@ne\hbox to\columnwidth{\hss\kern-2\captionindent\box\@ne\hss}%
  \else 
    \setbox\@ne\vbox{\unvbox\@tempboxa\parskip\z@skip
        \noindent\unhbox\@ne\advance\hsize-2\captionindent\par}%
  \fi
  \addvspace\abovecaptionskip
  \hbox to\hsize{\kern\captionindent\box\@ne\hss}%
\relax
}
\makeatother

\begin{document}

\title[Digital morphogenesis via Schelling segregation]
{Digital morphogenesis via Schelling segregation}

\author{George Barmpalias}
\author{Richard Elwes}
\author{Andy Lewis-Pye}
\thanks{An extended abstract of this paper with most proofs and some results omitted, has appeared in
 the Proceedings of the  IEEE Symposium on  
 Foundations of Computer Science (FOCS) 2014.  Authors are listed alphabetically. 
Barmpalias was supported by the 
1000 Talents Program for Young Scholars from the Chinese Government,
and the Chinese Academy of Sciences (CAS) President's International 
Fellowship Initiative No. 2010Y2GB03.
Additional support was received by
the CAS and the Institute of Software of the CAS.
Partial support was also received from a Marsden grant of New Zealand 
and the China Basic Research Program (973) grant No.~2014CB340302.
Lewis-Pye was supported by a Royal Society University 
Research Fellowship.}

\maketitle
\begin{abstract}  
Schelling's model of segregation looks to explain the way in which  particles or agents of two types may come to arrange themselves spatially into configurations consisting of large homogeneous clusters, i.e.\ connected regions consisting of only one type.  
As one of the earliest agent based models studied by economists and perhaps the most famous model of self-organising behaviour, it also has direct links to areas at the interface between computer science and statistical mechanics, such as the  Ising model and  the study of contagion and cascading phenomena in networks.  

While the model has been extensively studied it has largely resisted rigorous analysis, prior results from the literature generally pertaining to variants of the model which are tweaked so as to be  amenable to standard techniques from statistical mechanics or stochastic evolutionary game theory.    In \cite{BK}, Brandt, Immorlica, Kamath and Kleinberg provided the first  rigorous analysis of the  unperturbed model,  for a specific set of input parameters.  Here we provide a rigorous analysis of the model's behaviour much more generally and establish some  surprising forms of threshold behaviour, notably the existence of situations where an \emph{increased} level of intolerance for neighbouring agents of opposite type  leads almost certainly to \emph{decreased} segregation.  
\end{abstract}

\vspace*{\fill}
\noindent\textsc{{\bf George Barmpalias}}\\[0.5em]
\noindent
State Key Lab of Computer Science, 
Institute of Software, Chinese Academy of Sciences, Beijing, China.
School of Mathematics, Statistics and Operations Research,
Victoria University of Wellington, New Zealand.\\[0.2em] \textit{E-mail:} \texttt{barmpalias@gmail.com}.
\textit{Web:} \texttt{\href{http://barmpalias.net}{http://barmpalias.net}}\\[0.1em]\par
\addvspace{\medskipamount}
\noindent\textsc{{\bf Richard Elwes}}\\[0.5em]
\noindent School of Mathematics,
University of Leeds, LS2 9JT Leeds, United Kingdom.\\[0.2em]
\textit{E-mail:} \texttt{r.h.elwes@leeds.ac.uk.}
\textit{Web:} \texttt{\href{http://richardelwes.co.uk}{http://richardelwes.co.uk}} \\[0.1em]\par
\addvspace{\medskipamount}
\noindent\textsc{{\bf Andy Lewis-Pye}}\\[0.5em]  
\noindent Department of Mathematics,
Columbia House, London School of Economics, 
Houghton Street, London, WC2A 2AE, United Kingdom.\\[0.2em]
\textit{E-mail address:} \texttt{A.Lewis7@lse.ac.uk.}
\textit{Web:} \texttt{\href{http://aemlewis.co.uk/}{http://aemlewis.co.uk}} 
\vfill 
\thispagestyle{empty}
\clearpage

\section{Introduction} 

While Alan Turing is best known within the mathematical logic and computer science communities  for his work formalising the algorithmically calculable functions, it is interesting to note that  his most cited work \cite{AT} is actually that relating to morphogenesis. Turing wanted to understand certain biological processes: the gastrulation phase of embryonic development, the process whereby dappling effects arise on animal coats, and phyllotaxy, i.e.\ the arrangement of leaves on plant stems. One can consider the more general question, however, as to how morphogenesis occurs --  how structure can arise from an initially random, or near random configuration. Along these lines, one  of the major contributions of the economist and game theorist Thomas Schelling was an elegant model of segregation, first described in 1969 \cite{TS1},  which turns out to provide a very simple model of such a morphogenic process. This model looks to describe how individuals of different types come to organise themselves spatially into segregated regions, each of largely one type.  Today it has become perhaps the best known model of self-organising behaviour, and was one of the reasons cited by the committee upon awarding Schelling his Nobel memorial prize in 2005. 

For Schelling, part of the significance of his model was that it provided evidence for a recurrent theme in his research -- especially as elaborated upon in his influential work \cite{TS2} -- that individuals acting according to their local interests can produce global results which are undesired by all. Running small simulations of his model, for example, he observed that large levels of segregation can result in communities, when each individual has no preference for segregation,  but still requires a certain low proportion of their own local neighbourhood to be of their own type.  Although the explicit aim was initially to model the kind of racial segregation observed in large American cities, Schelling himself pointed out that the analysis is sufficiently abstract that any situation in which objects of two types arrange themselves geographically according to a certain preference not to be of a minority type in their neighbourhood, could constitute an interpretation. As pointed out in \cite{VK}, for example, Schelling's model can be seen as a finite difference version of  differential equations describing interparticle forces and applied in  modelling cluster formation. Many authors (see for example \cite{SS,DM,PW,GVN,GO}) have pointed out direct links to spin-1 models used to analyse phase transitions -- by introducing noise into the dynamics of the underlying Markov process one can arrive at the Boltzmann distribution for the set of possible configurations, with the `energy'  typically corresponding to  some measure of the mixing of types. From there one can immediately deduce that the modified (now ergodic) process spends a large proportion of the time in completely segregated states, with this proportion tending to 1 as the analogue of the temperature is taken to 0. So had Schelling been aware of these connections to variants of the Ising model, he could have based his work on a long history of physics research. 

Our own avenue into these questions, however, came via the work of computer scientists \cite{BK} and the study of cascading phenomena on networks as studied by Barabasi, Kleinberg and many others, a good introduction to which can be found in \cite{JK}.  The dynamics of the Schelling process (as will be clear once it has been formally defined below) are almost identical to many of those used to model the flow of information or behaviour on large social or physical networks such as the internet or the web of social contacts via which disease may spread --  the principal differences here being the initial conditions considered and the use of a much simpler underlying graph structure.  An immediate concern, given the results of this paper, is as to whether techniques developed here can be extended and applied to understand emergent phenomena on the various (normally more complex) random graph structures studied by those in the networks community.  Along these lines, Henry, Pra\l at and Zhang have described a simple but elegant model of network clustering \cite{HPZ}, inspired by the Schelling model. Their model doesn't display the kind of involved threshold behaviour, however,  that one might expect to be exhibited by a direct translation of the Schelling process to random underlying graph structures.

In \cite{CF} Clark and Fossett give an account of the role that  Schelling's model has had on the debate concerning anti-discrimination policy and the ongoing battle against residential segregation, along ethnic lines and others.   Certainly there has been increased governmental recognition of the fact that segregation has become one of the most important socio-political and public economic issues, with residential segregation having knock-on effects in terms of education, underachievement and disadvantages in the labour market \cite{TN}  and healthcare \cite{RB}.  In order to inform policy choices in combatting segregation it is essential to understand the extent to which this is an outcome  caused by voluntary residential preferences on the one hand, or by constraints on choices resulting from discrimination inherent  the system on the other.  Schelling's model (and more sophisticated variants, see for example \cite{JZ1})  have been used to argue that one might expect segregation to persist  even in the absence of systemic discrimination.  While the  simplicity of the model might initially suggest that one should be wary of rushing to draw real world conclusions, the robustness of the segregation phenomenon is certainly striking and occurs in all versions of the model which have been considered. Somewhat counter-intuitively, Zhang \cite{JZ2,JZ3} and Pancs and Vriend \cite{PV} have even shown segregation to result in variants of the model in which individuals \emph{actively prefer} integration. \\

  We concentrate here on the one-dimensional version of the model, as in \cite{BK}. The model works as follows. One begins with a large number $n$  of
nodes (individuals) arranged in a circle. Each node is initially assigned a \emph{type}, and has probability $\frac{1}{2}$ of being of type $\alpha$ and probability $\frac{1}{2}$ of being of type $\beta$ (the types of distinct individuals being independently distributed).  We fix a parameter $w$, which specifies the `neighbourhood' of each node in the following way: at each point in time the neighbourhood of the node $u$, denoted $\mathcal{N}(u)$,  is the set containing $u$ and the $w$-many closest neighbours on each side -- so the neighbourhood consists of $2w+1$ many nodes in total. The second parameter $\tau\in [0,1]$ specifies the proportion of a node's neighbourhood which must be of their type for them to be happy. So, at any given moment in time, we define $u$ to be \emph{happy} if at least $\tau(2w+1)$ of the nodes in $\mathcal{N}(u)$ are of the same type as $u$. 
One then considers a discrete time process, in which, at each stage, one pair of unhappy individuals of opposite types are  selected uniformly at random and are given the
opportunity to swap locations. Following \cite{BK}, we work according to the assumption that the swap will take place as long as each member of the pair has at least as many neighbours of the same type at their new location as at their former one (note that for $\tau\leq \frac{1}{2}$ this will automatically be the case). The process ends when (and if) one reaches a stage at which no further swaps are possible. \\

 Much of the difficulty  in providing a rigorous analysis stems from the large variety of absorbing states for the underlying Markov process. Various authors have therefore worked with variants of the model in which perturbations are introduced into the dynamics so as to avoid this problem, i.e.\ the model is  altered by introducing a further random  element, allowing individuals to sometimes make  moves which are detrimental with respect to their utility function. Such changes in the model dramatically simplify the dynamics, and might be justified by the assumption that we are dealing with individuals of  `bounded rationality' -- one might consider that precise information concerning the racial composition of each neighbourhood is not available, for example. As remarked above, these discussions are often couched in the language of statistical mechanics. In fact though, Young used techniques from evolutionary game theory -- an analysis in terms of stochastically stable states -- to develop the first results along these lines \cite{HY}, and these ideas were then substantially developed in a number of papers by Zhang \cite{JZ1,JZ2,JZ3}. While the language used may differ from that of the physicists, the basic analysis is essentially equivalent: Zhang establishes a Boltzmann distribution for the set of configurations, and then his stochastically stable states correspond to ground states. \\

   In \cite{BK}, Brandt, Immorlica, Kamath and Kleinberg  used an  analysis of  locally defined stable configurations, combined with results of Wormald \cite{NW}, to provide the first rigorous analysis of the unperturbed one-dimensional Schelling model, for the case $\tau=\frac{1}{2}$. The results obtained there are very different to those for the perturbed models: in the final configuration the average length of maximal segregated region is independent of $n$ and only polynomial in $w$. So now the local dynamics do not induce global segregation in proportion to the size of the society but instead only a small degree of segregation at the local level. The suggestion is that these results are in accord with empirical studies of residential segregation in large populations \cite{WC,FR,MW}. \\

  \textbf{Our contribution}.   In this paper we shall consider what happens more generally for $\tau \in [0,1]$ (for the unperturbed model), and we shall observe, in particular, that some remarkable threshold behaviour occurs.  While some aspects of the approach from \cite{BK} remain, in particular the focus on locally defined stable configurations which can be used to understand the global picture, the specific methods of their proof (the use of `firewall incubators', and so on) apply only to the case  $\tau=\frac{1}{2}$, and so largely speaking we shall require different techniques here.   The picture which emerges is one in which one observes different behaviour in five regions. For $\kappa$ which is the unique solution in $[0,1]$ to: 
   
   \[ \left( \frac{1}{2}-\kappa \right)^{1-2\kappa} = \left( 1-\kappa \right)^{2-2\kappa}, \] 
   
    ($\kappa \approx 0.353092313$) these regions are: \begin{inparaenum}[(i)] \item $\tau<\kappa$, \  \item $\tau=\kappa$, \  \item  $\kappa<\tau<\frac{1}{2}$, \  \item  $\tau= \frac{1}{2} $, \  \item  $\tau >\frac{1}{2}$ \end{inparaenum}. In fact we shall not consider the case $\tau =\kappa$, but the behaviour for all other values of $\tau$ is given by the theorems below. Perhaps the most surprising fact is that, in some cases,  increasing $\tau$ almost certainly leads to decreased segregation. The assumption is always that we work with $n\gg w$, i.e.\ all results hold for all $n$ which are sufficiently large compared to $w$.   A \emph{run} of length $d$ is a set of $d$-many consecutive nodes all of the same type. \emph{Complete segregation} refers to any configuration in which there exists a single run to which  all $\alpha$ nodes belong. 
   
    The first theorem deals with low values of $\tau$, and formally establishes the (perhaps rather intuitive)  idea that very low levels of intolerance lead to low levels of segregation:  
   
   \begin{thm} \label{1}
   Suppose  $\tau<\kappa$ and $\epsilon>0$. For all sufficiently large $w$, if a node $u$ is chosen uniformly at random, then the probability that any node in $\mathcal{N}(u)$ is ever involved in a swap is $<\epsilon$.  
Thus there exists a constant $d$ such that, for sufficiently large $w$,  the probability $u$ belongs to a run of length $>d$ in the final configuration is $<\epsilon$.\footnote{A comment on conventions concerning quantification may be in order here. In the  statement of Theorem \ref{1} we fix $\tau$ and $\epsilon$ and then conclude that, for all sufficiently large $w$,  a certain statement holds. It is to be understood here that how large one has to take $w$ may depend upon $\tau$ and $\epsilon$. Similarly, in Theorem \ref{3}, how large one has to take $w$, and then how large $n$ must be in comparison to $w$, may depend upon $\tau$ and $\epsilon$. The upshot of this is that, where we subsequently make multiple uses of the weak law of large numbers, there shall be no requirement of \emph{uniform} convergence for the random variables in question.} 
   \end{thm}

 As one increases $\tau$ beyond the threshold $\kappa$, however, the dynamics of the process qualitatively change:  
    
   \begin{thm} \label{3} Suppose $\tau \in \left( \kappa,\frac{1}{2} \right)$ and  $\epsilon>0$. 
   There exists a constant $d$ such that 
   (for all $w$ and all $n$ such that $w\ll n$) the probability that $u$ chosen uniformly at random will belong to a run of length $\geq e^{w/d}$ in the final configuration, is greater than $1-\epsilon$.
      \end{thm} 
      
        So, to summarise Theorem \ref{3} less formally, increasing $\tau$ beyond $\kappa$ suddenly causes high levels of  segregation, in the form of run-lengths which are exponential in $w$. Furthermore, by analysing the proof of Theorem \ref{3}, we shall be able to prove (a formalised version of the statement) that  increasing $\tau$ in this interval actually decreases  run-lengths.   
      The next case is that dealt with in \cite{BK}, where one sees polynomially bounded run-lengths:

      \begin{thm}[\cite{BK}] Suppose $\tau=\frac{1}{2}$. There exists a constant $c < 1$  such that for all $\lambda > 0$, the probability that $u$ chosen uniformly at random will belong to a run of length greater than $\lambda w^2$ in the final configuration,  is bounded above by $c^{\lambda}$. 
      \end{thm}
       
   The final case is when $\tau>\frac{1}{2}$. Here a combinatorial argument can be used to argue that complete segregation will eventually occur. Note that, if $\frac{1}{2}<\tau\leq   \frac{w+1}{2w+1}$, then the process is identical to that for $\tau=\frac{1}{2}$, since in both cases a node requires $ w+1$ many nodes of its own type in its neighbourhood  in order to be happy. 
       
   \begin{thm} \label{5}  Suppose that $\tau>\frac{1}{2}$, and that $w$ is sufficiently large that $\tau>
   \frac{w+1}{2w+1}$. Then, with probability tending to 1 as $n\rightarrow \infty$, the initial configuration is such that complete segregation is inevitable.\footnote{Note that, while Theorem \ref{5} requires   $\tau> \frac{w+1}{2w+1}$, Theorem \ref{3} does not. Briefly, this is because the choice of $d$ can be made to so as to deal with finitely many small $w$ anyway, meaning that Theorem \ref{3} is essentially a statement about what happens  for large $w$.}  
   \end{thm}

 \begin{figure} \label{firsta} 
\includegraphics[scale=0.6, clip =true, trim=4cm 0.5cm 4cm 0.5cm]{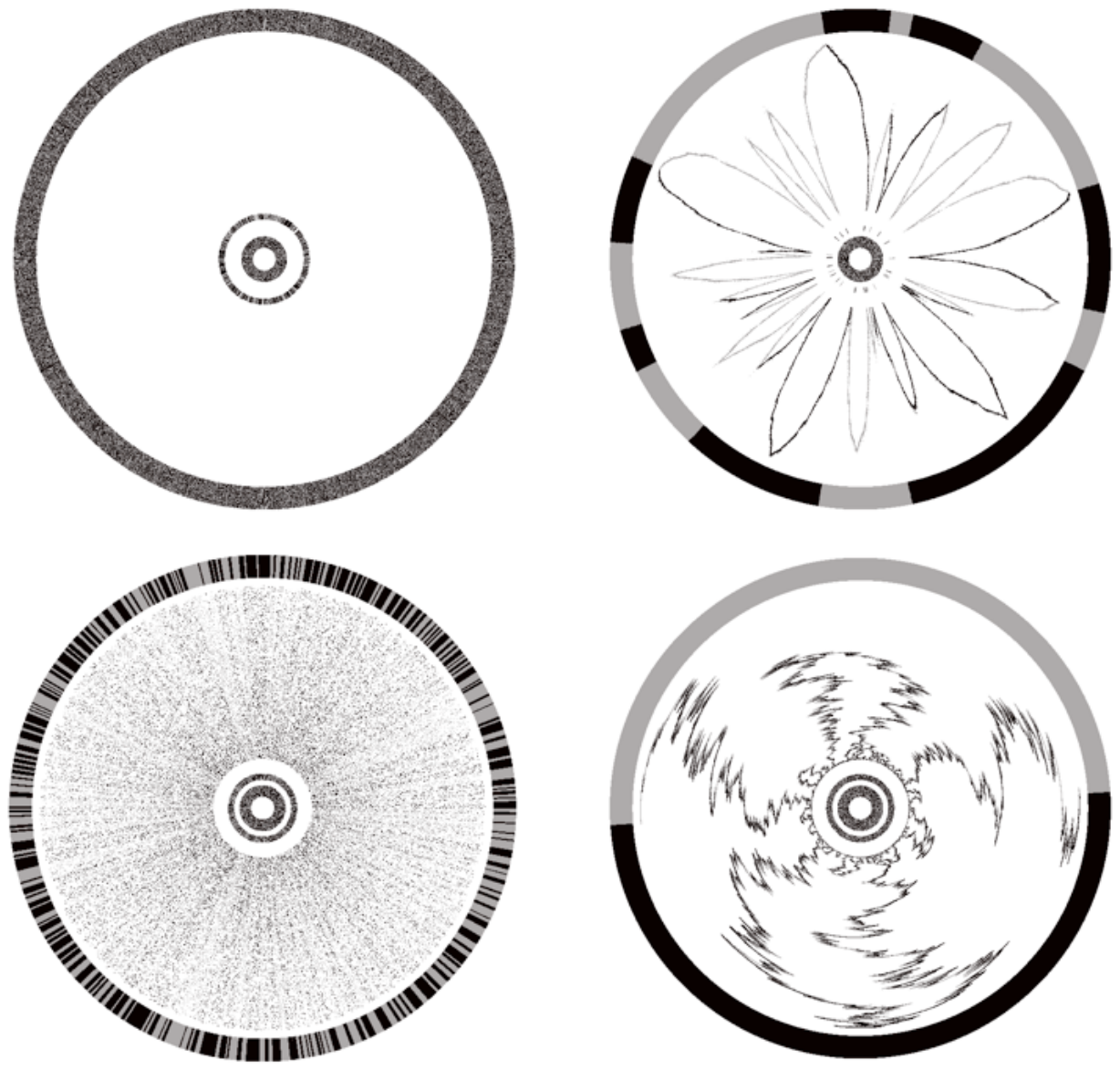}
 \caption{Threshold behaviour: top left $\tau=0.3$; top right $\tau= 0.38$; bottom left $\tau= 0.5$; bottom right $\tau =0.7$. }
\label{fig:thresbehav}
\end{figure}  

We have constructed a program which efficiently simulates the process. The outcomes of some simulations are  illustrated in Figure \ref{fig:thresbehav}. In the processes depicted here the number of nodes $n=100000$, $w=60$ and in the diagrams individuals of type $\alpha$ are coloured light grey and individuals of type $\beta$ are coloured black (much larger simulations are also illustrated in Section \ref{half}).  The inner ring displays the initial  mixed configuration (in fact the configuration is sufficiently mixed that changes of type are not really visible, so that the inner ring appears dark grey). The outer ring displays the final configuration. Just immediately exterior to the innermost ring are second and third inner rings, which display individuals which are unhappy in the initial configuration and individuals belonging to `stable' intervals in the initial configuration respectively (stable intervals will be defined subsequently, and in fact, this third ring is empty in these examples except for the case $\tau=0.3$).  The process by which the final configuration is reached is indicated in the space between the inner rings and the outer ring in the following way: when an individual changes type this is indicated with a mark,  at a distance from the inner rings which is proportional to the time at which the change of type takes place. In fact, for the case $\tau> \frac{1}{2}$ one has to be a little careful in talking about the `final' configuration -- there will, almost certainly, always be unhappy individuals of both types able to swap, but once a completely segregated configuration is reached all future configurations must remain completely segregated.     \\
 
 We shall also be interested in a variant of the model, which we shall call the \emph{simple model} (also see \cite{BEL}), and which proceeds in the same way except that at each stage \emph{one} unhappy node $u$ is selected uniformly at random and changes type so long as this does not cause it to have less nodes of its own type within $\mathcal{N}(u)$. The process ends when no more legal changes are possible. One might justify interest in this version of the model in various ways. Firstly there are a number of situations in which it is much easier  to work with (in one instance here, and also in \cite{BK}, results are proved for  the standard model by first considering the simple model and then arguing that the proof can be extended to the standard case). The simple variant of the model also makes sense as soon as one drops the assumption that we are working within a closed system. One might suppose that unhappy individuals living in a city will move to a location in the same city if one should be available, but will move elsewhere otherwise, and similarly that individuals will move into the city to fill locations becoming vacant. Lastly, this simple variant of the model is also closer to the various spin-1 models normally studied in statistical mechanics (although versions of the Ising model with `Kawasaki' dynamics involving particles which swap location are also considered there).   In all cases we shall prove the same results for both models, except for the case $\tau>\frac{1}{2}$, where the  simple model eventually yields a society of nodes all of one type.  The rough conjecture is that the simple model accurately describes local behaviour for the standard model.  \\
 
 In Section \ref{notation} we shall discuss terminology and we shall make some easy observations about how the evolution of the model can be understood. In Section \ref{taulessthankappa} we prove Theorem \ref{1}. In Section \ref{hardtau} we prove Theorem \ref{3}. In Section \ref{half} we discuss the case $\tau=\frac{1}{2}$, which was already dealt with in \cite{BK}, and in  Section \ref{complete} we prove Theorem \ref{5}.  Sections \ref{def1}, \ref{def2} and \ref{def3} deal with deferred proofs from Sections   \ref{taulessthankappa}, \ref{hardtau} and  \ref{complete} respectively. 
  
\section{Some notation, terminology and some easy observations} \label{notation} 

In describing the model earlier, we talked in terms of nodes or individuals swapping locations at various stages of the dynamic process. To work in this way, however, requires one to draw a distinction between individuals and the locations they occupy at each stage, and to maintain a bijective map at each stage between these two sets. In fact, it is notationally easier  to consider a process whereby one simply has a set of $n$ nodes, with two unhappy nodes of opposite type selected at each stage (if such exist), which may then both change type (when this occurs we shall still refer to the nodes as `swapping', but now they are swapping type rather than location). 
Thus nodes are identified with indices for their locations amongst the set $\{ 0,1,...,n-1 \}$, and unless stated otherwise, addition and subtraction on these indices are performed modulo $n$. In the context of discussing a node $u_1$, for example,  we might refer to the immediate neighbour on the right as node $u_1+1$.  Since we work modulo $n$ it is worth clarifying some details of the interval notation: for  $0\leq b<a<n$, we let $[a,b]$ denote the set of nodes (`interval')  $[a,n-1] \cup [0,b]$ (while $[b,a]$ is, of course, understood in the standard way). 

As noted before, for any node $u$, we let $\mathcal{N}(u)$ denote the neighbourhood of $u$, which is  the interval $[u-w, u+w]$. For any set of nodes $I$, suppose that $x$ is the number of $\alpha$ nodes in $I$, while $y=|I|-x$. Then $\Theta(I):=x-y$ and is called the \emph{bias} of $I$. By the bias of a node we mean the bias of its neighbourhood.  Recall that by a \emph{run} of length $m+1$ we mean an interval $[u,u+m]$ in which all nodes are of the same type. 

We shall be particularly interested in  local configurations which are stable, in the sense that certain nodes in them  can never be caused to change type. Note that if an interval of length $w+1$ contains at least $\tau (2w+1)$ many $\alpha$ nodes, then each of those $\alpha$ nodes is happy so long as the others do not change type, meaning that, in fact, no $\alpha$ nodes in that interval will ever change type. We say that such an interval of length $w+1$ is $\alpha$-stable (and similarly for $\beta$). An interval of length $w+1$ is \emph{stable} if it is either $\alpha$-stable or $\beta$-stable. We shall also make use of a particular kind of stable interval which was used in \cite{BK}:   a \emph{firewall}  is a run of length at least $w+1$. 
We write `for $0\ll w\ll n$', to mean `for all sufficiently large $w$ and all $n$ sufficiently large compared to $w$'. \\

 \textbf{Arguing that the process ends.} We define the \emph{harmony index} corresponding to any given configuration to be the sum  over all nodes of the number of  their own type within their neighbourhood. For $\tau \leq \frac{1}{2}$, this harmony index is easily seen to strictly increase whenever an unhappy node changes type, which combined with the existence of an upper bound $n(2w+1)$, implies that the process must terminate after finitely  many stages. For $\tau>\frac{1}{2}$, we shall argue that with probability tending to 1 as $n\rightarrow \infty$, the initial configuration is such that complete segregation eventually occurs with probability 1. Once complete segregation has occurred it is easy to see that all future states must be completely segregated, but that `rotations' can occur, i.e.\ if the nodes of type $\alpha$ are precisely the interval $[a,b]$ at stage $s$, then at stage $s+1$ they must be either $[a-1,b-1],[a,b]$ or $[a+1,b+1]$.

\section{The case \texorpdfstring{$\tau<\kappa$}{tau less than kappa}} \label{taulessthankappa}

The analysis here is identical for the standard and simple models. Of course, we are yet to explain how $\kappa$ comes to be defined as described previously  -- we shall do so shortly. 

 \textbf{Outline.} The basic idea is that we wish to find $\kappa$, which is that value of $\tau$ at which stable intervals become more likely than unhappy nodes in the initial configuration. For such $\tau$, taking $w$ large, we shall have that stable intervals are \emph{much} more likely than unhappy nodes in the initial configuration. With a little bit more work one can then show that, for a node  $u$ which is selected uniformly at random, we shall almost certainly find stable intervals of both types on either side of $u$ before any unhappy element. The following lemma then suffices to establish that, given such an initial configuration, $u$ can never change type. 

\begin{lem} \label{stopsem} 
Suppose that, in the initial configuration,  $u_1$ and $u_2$ each belong to (possibly different) $\alpha$-stable intervals, and that there are no unhappy $\alpha$ nodes in $[u_1,u_2]$.  Then no $\alpha$ node in $[u_1,u_2]$ will ever become unhappy. A similar result holds for $\beta$. 
\end{lem} 
\begin{proof} 
Suppose otherwise, and let $v$ be the first $\alpha$ node in the interval $[u_1,u_2]$ to become unhappy. In order for $v$ to become unhappy, another $\alpha$ node $v'\in \mathcal{N}(v)$ must change to type $\beta$. Since $v'\notin [u_1,u_2]$, we either have $v'<u_1 \leq v$, or else $v\leq u_2 <v'$. Suppose that the first case holds, the other is similar. Then, together with the fact that any $\alpha$-stable interval to which $u_1$ belongs is of length $w+1$, $v'\in \mathcal{N}(v)$ implies that $v$ belongs to any stable interval to which $u_1$ belongs, and so cannot become unhappy. This  gives the required contradiction.  
\end{proof}

 \textbf{Finding} $\kappa$. We are interested in the probability that a randomly chosen node belongs to a $\beta$-stable interval in the initial configuration.  As a first approximation, however, we begin by asking the following slightly simpler version of that question: given some node $u$, chosen uniformly at random and of type $\beta$, say, what is the probability that $[u,u+w]$ is $\beta$-stable in the initial configuration? This can be modelled with a binomial distribution $X \sim b(w, \frac{1}{2})$, from which we are interested in the probability $\Ps = \textbf{P}(X \geq (2w+1) \tau -1)$.

Similarly, an $\alpha$ node is unhappy if its neighbourhood contains more than  $(2w+1)(1 - \tau)$ many $\beta$ nodes. We model this as  $Y \sim b(2w, \frac{1}{2})$, from which we are interested in the probability $\Pu = \textbf{P}(Y > (2w+1)(1 - \tau))$. We are interested in finding conditions on $w$ and $ \tau$ which ensure: 
  $$\frac{\Pu}{\Ps}<1 \ \ \mbox{or}\ \ \frac{\Pu}{\Ps}< \eps.$$

First, notice that if $\tau \leq \frac{1}{4}$, then $\Ps \geq \frac{1}{2}$ for all $w$, while $\Pu \to 0$ as $w \to \infty$. Thus for all large enough $w$, we shall have $\frac{\Pu}{\Ps}< \eps$.  Thus, as we look to find $\kappa$,  we shall assume that $\frac{1}{4} < \tau < \frac{1}{2}$. Our key probabilistic fact will be the following result concerning the binomial distribution, which is presumably known, and which is proved in Section \ref{def1}:

\begin{lem} \label{lem:binom}
Suppose $h: \mathbb{N} \to \mathbb{N}$ and $p \in (0,1)$ are such that there exists $k \in (0,1)$ so that for all large enough $N$, we have $\left( 1+ \left( \frac{1}{p} -1 \right) k \right) h(N)> N \geq h(N)> pN > 0$. Then for all large enough $N$, if $X_N \sim b(N,p)$, we have 
$$P \left( X_N = h(N) \right) \ \ \leq \ \  \textbf{P} \left( X_N \geq h(N) \right) \ \ \leq \ \ \left( \frac{1}{1-k} \right) \cdot \textbf{P} \left( X_N = h(N) \right).$$
That is to say in asymptotic notation, $\textbf{P} \left( X_N \geq h(N) \right)= \Theta \left( \textbf{P} \left( X_N = h(N) \right) \right)$.
\end{lem}

Applying Lemma \ref{lem:binom} in our setting with $p=\frac{1}{2}$, $N=w$, and $1>k> \frac{\frac{1}{2}- \tau}{\tau}$ gives us that $\ds \Ps \approx \frac{1}{2^w} \bpm w \\ h \epm $ where $h=\lceil (2w+1) \tau \rceil-1$ and where by ``$\approx$'' we mean the asymptotic notion $\Theta$.

Similarly, taking $N=2w$ and $1>k'>\frac{\tau}{1- \tau}$ gives us $\ds \Pu \approx \frac{1}{2^{2w}} \bpm 2w \\ h' \epm$, where $h'=\lfloor (2w+1)(1 - \tau))\rfloor +1$.
 Thus  $$\frac{\Pu}{\Ps} \approx \frac{1}{2^w} \cdot \frac{\bpm 2w \\ h' \epm}{\bpm w \\ h \epm}.$$

 We now employ Stirling's formula, that $n! \approx n^{n + \frac{1}{2}} e^{-n}$. Then, the powers of $e$ cancel and we see:

\begin{equation} \label{R2} \frac{\Pu}{\Ps} \approx \frac{1}{2^w} \cdot  \frac{(2w)^{2w+\frac{1}{2}}(h)^{h+\frac{1}{2}}(w-h)^{w-h+\frac{1}{2}}}{w^{w+\frac{1}{2}}(h')^{h'+\frac{1}{2}}(2w-h')^{2w-h'+\frac{1}{2}}}. \end{equation}

 Now, approximating $h$ by $2w \tau$ and approximating $h'$ by $ 2w(1-\tau)$ (which we shall justify shortly) we get:

$$ \frac{\Pu}{\Ps} \approxeq {(2w)^w} \cdot \frac{(2w \tau)^{2w \tau +\frac{1}{2}}\left(w(1-2 \tau) \right)^{w(1- 2 \tau) + \frac{1}{2}} }{(2w(1-\tau))^{2w(1-\tau)+\frac{1}{2}} (2w \tau)^{2w \tau +\frac{1}{2}}}.$$

and so \begin{equation} \label{R3} \frac{\Pu}{\Ps} \approxeq \left( \frac{ \left(\frac{1}{2}- \tau \right) ^{(1- 2 \tau)} }{(1-\tau)^{2(1-\tau)}} \right)^w. \end{equation}

 Here, we write $f\approxeq g$ to mean that there are rational functions $P$ and $Q$ such that $P(w),Q(w)>0$ and $P(w)g(w) \leq f(w) \leq Q(w)g(w)$.

\begin{figure}
\includegraphics[width=3.4in]{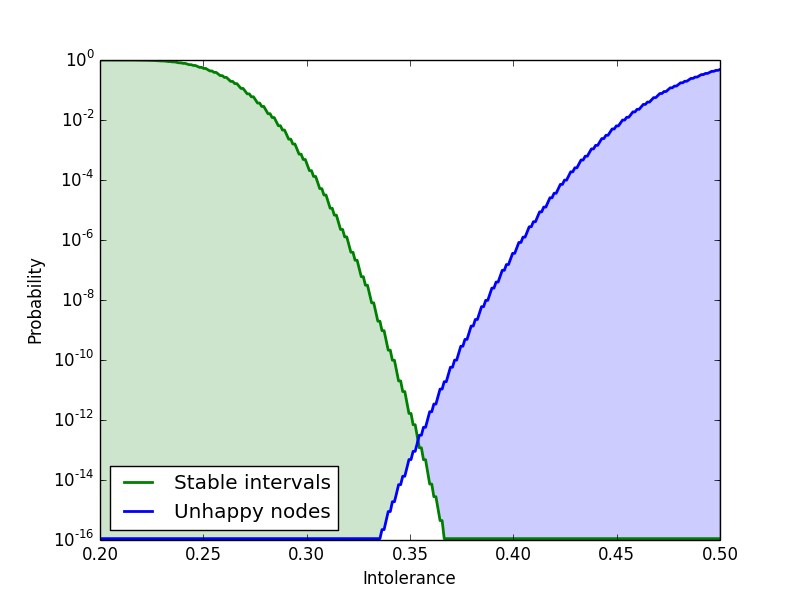}
\caption{Probabilities (log scale) for $w=300$.}
\label{fig:sample_figure}
\end{figure}

Now put $f(x)=x^{2x}$. Then the value $ \kappa$ we are looking for is  exactly $\tau$  such that
\[  f \left( \frac{1}{2}-\tau \right)=f(1-\tau).\]  In other words  $x=\frac{1}{2}-\kappa$ is exactly such that
\[  f(x)=f\left( x+\frac{1}{2} \right).\]  Since $f$ has a unique turning point at $e^{-1}$, it follows that  $\kappa$ is unique, and   numerical analysis gives \[ \kappa \approx 0.353092313\]  
 (which is just slightly less than $\sqrt{2}/4$).

It remains to justify the approximations introduced for $h$ and $h'$ after \ref{R2}. Well, $h = 2w \tau + \delta$ where $-1 < \delta < 1$. Thus $h^{h+\frac{1}{2}} = (2w \tau + \delta)^{\delta +\frac{1}{2}} \cdot  (2 w \tau)^{2 w \tau} \cdot \left(1 + \frac{\delta}{2w \tau} \right)^{2w \tau}$. The final term lies within $[e^{-1}, e]$. Thus the necessary correction amounts to multiplying or dividing $(2w \tau)^{2w \tau +\frac{1}{2}}$ by at most a polynomial term in $w$. Similar small polynomial corrections may be required for the other terms in \ref{R2}, but the exponential nature of \ref{R3} guarantees that these correcting terms will play a negligible role for large $w$.\\

There were various other simplifications which were made in the  analysis above, and which must now be attended to.  In defining $\Ps$, we considered the probability that a $\beta$ node belongs to a $\beta$-stable interval, rather than the probability that a node chosen uniformly at random belongs to a $\beta$-stable interval. 
We also considered a  specific interval of length $w+1$ to which a node might belong and which might be stable, namely $[u,u+w]$, and so did not overtly take account of the fact that a given node belongs to $w+1$ many intervals of length $w+1$ which might be stable, namely $[u-\ell,u-\ell+w]$ for $0 \leq \ell \leq w$. Again, the exponential nature of (\ref{R3}) accommodates all such modifications, which only cause changes to $\Pu$ and $\Ps$ which are are polynomially bounded in $w$. While (\ref{R3}) was derived under the assumption that $\frac{1}{4}<\tau<\frac{1}{2}$, in fact we can deduce that the following statement holds more generally: 

\begin{enumerate} 
\item[($\dagger_0$)]  Consider the initial configuration. For $u$ chosen uniformly at random, let $\Ps$ be the probability that $u$ belongs to an $\alpha$-stable interval, and let $\Pu$ be the probability that  $u$ is an unhappy $\alpha$ node.  If  $\tau <\kappa$ and $k,r>0$, then for all sufficiently large $w$, $\Ps> kw^r\cdot \Pu$. When $\tau >\kappa$, $\Pu> kw^r\cdot \Ps$. A similar result holds for $\beta$. \end{enumerate}

 In order to see that $(\dagger_0)$ holds outside the interval $\frac{1}{4}<\tau <\frac{1}{2}$, (for the non-trivial cases) one can simply compare the relevant probability ratios with those for $\tau$ just inside the interval. \\

 \textbf{Completing the argument}. As hinted previously, we need something more than $(\dagger_0)$  though. For any $u$, let $x_u$ be the first node to the left\footnote{By the first node to the left of $u$ satisfying a certain condition we mean the first in the sequence $u,u-1,u-2,\cdots$  which satisfies the condition.} of $u$ which, in the initial configuration, is either an unhappy $\alpha$ node, or else belongs to an $\alpha$-stable interval. We must show that when $\tau<\kappa$ and $w$ is sufficiently large, it is almost certainly the case that $x_u$ belongs to an $\alpha$-stable interval (and similarly for $\beta$ and to the right of $u$).  With this in place, the result then follows directly from Lemma \ref{stopsem}. 
 When combined with $(\dagger_0)$, the following lemma (which is stated in such a way as to be general purpose, so we can apply it again later) therefore completes our proof, and is proved in Section \ref{def1}.\footnote{For the reader who is concerned that we need to make four applications of Lemma \ref{gen} to various different events $P_u$ and $Q_u$ which are not independent (with $P_u$ either being the event that $u$ belongs to an $\alpha$-stable interval or that $u$ belongs to a $\beta$-stable interval, and with $Q_u$ either being the event that $u$ is an unhappy $\alpha$ node or the event that it is an unhappy $\beta$-node), note that the probability that any one of four unlikely events occurs is at most the sum of their individual probabilities.}

\begin{lem} \label{gen}  

Let $P_u$ and $Q_u$ be events which only depend on the neighbourhood of $u$ in the initial configuration, meaning that if the neighbourhood of $v$ in the initial configuration is identical that of  $u$ (i.e. for all $i \in [-w,w]$, $u+i$ is of the same type as $v+i$), then $P_u$ holds iff $P_v$ holds and $Q_u$ holds iff $Q_v$ holds. Suppose also that:
\begin{enumerate}[(i)]  
\item $\textbf{P}(P_u)\neq 0$ and $\textbf{P}(Q_u)\neq 0$. 
\item  For all $k$, for all sufficiently large $w$, $\textbf{P}(P_u)/\textbf{P}(Q_u) >kw$.
\end{enumerate}   
For any $u$, let $x_u$ be the first node to the left  of $u$ such that either $P_{x_u}$ or $Q_{x_u}$ holds.  For any $\epsilon>0$, if $0\ll w \ll n$ then the following occurs with probability $>1-\epsilon$ for $u$ chosen uniformly at random:  $x_u$ is defined and for no node $v$  in $[x_u-2w,x_u]$ does $Q_v$ hold.

An analogous result holds when `left' is replaced by `right'. 
\end{lem}

\section{The case \texorpdfstring{$\kappa <\tau <\frac{1}{2}$}{tau between kappa and 1/2}} \label{hardtau}   
We work first with the simple model, and then supply the necessary modifications for the standard case. 

In what follows we shall work with some fixed  $\tau$ in the interval $\left( \kappa,\frac{1}{2} \right)$, some fixed $\epsilon>0$,  and we shall assume that $n$ is large compared to $w$. We want to show that there exists a constant $d$ such that for all sufficiently large $w$ the probability that a randomly chosen node  will belong to a run of length $\geq e^{w/d}$ (in the final configuration) is greater than $1-\epsilon$. Of course, proving the result for all sufficiently large $w$ suffices to give the result for all $w$ since one can simply adjust the choice of $d$ to deal with finitely many small values, but we shall make frequent use of the fact that we need only work for all sufficiently large $w$ in what  follows and so stating the theorem in this way is instructive. \\

Our entire analysis takes place relative to a node $u_0$, chosen uniformly at random. Roughly, the aim is to establish that in the final configuration $u_0$ very probably belongs to a firewall of considerable length. The argument consists essentially of two main parts: first we consider what can be expected from the vicinity\footnote{To be clear, the term `vicinity' of $u_0$ is used informally here, to mean an interval containing $u_0$, which may be large compared to $w$ but which is small compared to $n$.} of $u_0$ in the initial configuration, and then we consider how events are likely to develop in subsequent stages. 

Before we begin with the technicalities,  let us consider very informally what  can be expected from the vicinity of $u_0$ in the initial configuration. Since $\kappa<\tau<\frac{1}{2}$, for large $w$ we shall have that unhappy nodes are much more likely than stable intervals, but that unhappy nodes themselves are few and far between. Starting at $u_0$ and moving to the left, (since $n$ is large) we can expect to find a first unhappy node, $l_1$ say, and it will very likely be the case that $[l_1,u_0]$ is an interval of considerable length, containing no stable intervals. To the left of $l_1$ and inside $\mathcal{N}(l_1)$, there may be some other unhappy nodes. If we move now to $l_1-(2w+1)$, however,  and repeat the process (with $l_1-(2w+1)$ taking the place of $u_0$),   then so long as $w$ is large enough, we can expect the same to happen again, i.e.\ we find a first unhappy node $l_2$ and  $[l_2,l_1-(2w+1)]$ is an interval of considerable length, containing no stable intervals. In fact, for any fixed $k$ which does not depend on $w$, if we repeat this process $k$ many times, then so long as $w$ is large enough (and how large we have to take $w$ will depend on $k$) we can be pretty sure that the same thing will happen at every one of those $k$-many steps. Now let us establish this informal picture more carefully.

\subsection*{The initial configuration} Recall that for any set of nodes $I$,  $\Theta(I)$ is the  \emph{bias} of $I$ and that by the bias of a node we mean the bias of its neighbourhood.  In general, if $x_1,...,x_k$ are independent random variables with $\textbf{P}(x_i=1)=\textbf{P}(x_i=-1)=\frac{1}{2} $ when $1\leq i \leq k$, then letting $X= \sum_{i=1}^k x_i$ Hoeffding's inequality gives,  for arbitrary $\lambda>0$ : 

\[ \textbf{P}(|X|>\lambda  \sqrt{k} ) <2e^{-\lambda^2/2}. \] 

Now we use this to bound the probability that a node $u$ has  bias in the initial configuration which will cause it to be unhappy, should $u$ be of the minority type in its neighbourhood. So, we wish to bound the probability that the number of $\alpha$ nodes in $\mathcal{N}(u)$ is $>(1-\tau ) (2w+1)$ or the   number of $\beta$ nodes in $\mathcal{N}(u)$ is $>(1-\tau ) (2w+1)$. This corresponds to a bias $\Theta(\mathcal{N}(u))$ of $> (1-2\tau ) (2w+1) $ or $<-(1-2\tau ) (2w+1) $.

\begin{defin} 
  When $ | \Theta(\mathcal{N}(u))|> (1-2\tau ) (2w+1) $ we say that $u$ has \textbf{high bias}, denoted $\mathtt{Hb}(u)$.  If this holds for the initial configuration, we say that $\mathtt{Hb}^{\ast}(u)$ holds. 
\end{defin}  

\begin{defin} \label{d} 
For the remainder of this section we define $d= \frac{2}{(1-2\tau)^2} $. 
\end{defin}

\begin{lem}[Likely happiness]  \label{unhappy}
Let $u$ be a node chosen uniformly at random. For any $\epsilon'>0$ and for all sufficiently large $w$, the probability that $\mathtt{Hb}^{\ast}(u)$ holds is $<\epsilon' e^{-w/d}$. 
\end{lem}
  
\begin{proof}  Putting $\lambda \sqrt{2w+1} =  (1-2\tau ) (2w+1) $ in  Hoeffding's inequality above,  we get

\[ \lambda^2/2 =(1-2\tau)^2\frac{(2w+1)}{2}, \ \ \mbox{so}\ \  e^{-\lambda^2/2}= e^{-w/d^{\ast}}\ \ \ \mbox{where}\ \ \ \ d^{\ast}=\frac{2w}{2w+1}\cdot \frac{1}{(1-2\tau)^2}.\]

 We chose $d>\frac{1}{(1-2\tau)^2}$, which means that for any $\epsilon'>0$ and for all sufficiently large $w$, the probability that $u$ has high bias in the initial configuration is $<\epsilon' e^{-w/d}$, as required.
\end{proof}

 \textbf{Defining the nodes} $l_i$ \textbf{and} $r_i$. For now, we fix some $k_0>0$. We shall choose a specific value of $k_0$ which is appropriate later -- for now, however, we make the promise that our choice of $k_0$ will not depend on $w$.  

For $1\leq i \leq k_0$  we define a node   $l_i$ to the left of $u_0$ and also a node $r_i$ to the right.  We let $l_1$ be the first node $v$ to the left of $u_0$ such that $\mathtt{Hb}^{\ast}(v)$ holds, so long as this node is in the interval $[u_0-\frac{1}{4}n,u_0]$ (otherwise $l_1$ is undefined).  Then, given $l_i$ for $i<k_0$ we let $l_{i+1}$ be the first node $v$ to the left of $l_i-(2w+1)$  such that $\mathtt{Hb}^{\ast}(v)$ holds, so long as no nodes in the interval $[l_{i+1},l_i]$ are outside the interval  $[u_0-\frac{1}{4}n,u_0]$ (otherwise $l_{i+1}$ is undefined). 
We let $r_1$ be the first node $v$  to the right of $u_0$  such that $\mathtt{Hb}^{\ast}(v)$ holds, so long as this node is in the interval $[u_0,u_0+\frac{1}{4}n]$. Given $r_i$ for $i<k_0$ we let $r_{i+1}$ be the first node $v$ to the right  of $r_i+(2w+1)$  such that $\mathtt{Hb}^{\ast}(v)$ holds, so long as no nodes in the interval $[r_i,r_{i+1}]$ are outside the interval $[u_0,u_0+\frac{1}{4}n]$. 

 The reason for considering the intervals $[u_0-\frac{1}{4}n,u_0]$ and $[u_0,u_0+\frac{1}{4}n]$ in the above, is that we wish to be able to move left from $u_0$ to $l_{k_0}$ without meeting any of the nodes $r_i$. \\

\begin{figure}
\epsfig{figure=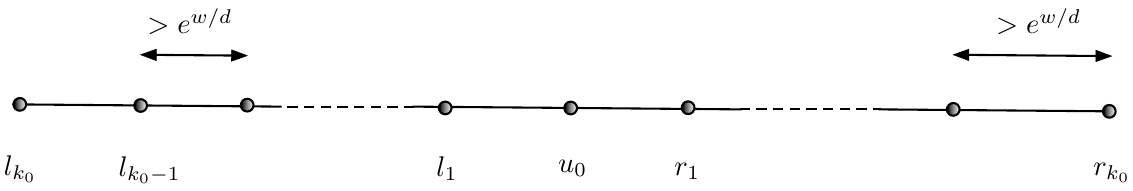,width=11cm}
\caption{Picking out nodes of high bias in the vicinity of $u_0$.}
\label{fig:definingthel_i}
\end{figure}

\begin{defin}[\textbf{Good spacing}] \label{Q0}
Let $d$ be as in Definition \ref{d}. It is notationally convenient to let $l_0=u_0=r_0$.   We let $\mbox{\textbf{Good spacing}}$   be the event that for  $1\leq i \leq k_0$:
\begin{enumerate} 
\item[(i)]  $l_i$ and $r_i$ are both defined, and;
\item[(ii)] $|l_i-l_{i-1}|>e^{w/d}$ and $|r_i-r_{i-1}|>e^{w/d}$. 
\end{enumerate} 
\end{defin} 

Note that for any fixed $w$, as $n\rightarrow \infty$ the probability that any $l_i$ or $r_i$ is undefined (for $i\leq k_0$) goes to 0. 
By Lemma \ref{unhappy}, and since the probability that any node in an interval $I$ has high bias in the initial configuration is at most $\Sigma_{u\in I} \textbf{P}(\mathtt{Hb}^{\ast}(u))$, for any $\epsilon'>0$ and for any fixed $k_0\geq 1$  we can ensure that $\textbf{P}(\mbox{\textbf{Good spacing}})>1-\epsilon'$ by taking $w$ sufficiently large (and by taking $n$ sufficiently large compared to $w$). Thus, for $0\ll w \ll n$,  the picture we are presented with is almost certainly as in Figure \ref{fig:definingthel_i}.\\


\textbf{Building the informal picture}. Recall that, by a \emph{run} of length $m+1$ we mean an interval $[u,u+m]$ in which all nodes are of the same type, and that a \emph{firewall} is a run of length at least $w+1$. The basic observation on which we now wish to build is as follows:  if the interval $[u-w,u]$ is a firewall of type $\alpha$, then when $u+1$ is of type $\beta$, it cannot be happy unless the interval $[u+1,u+w+1]$ is $\beta$-stable. So, since we are dealing with $\tau\leq \frac{1}{2}$: 

\[ \parbox{11cm}{Firewalls will spread until they hit stable intervals of the opposite type.\footnote{When dealing with the standard rather than the simple model, this is only true so long as there exist unhappy nodes of both types. Later, our analysis using Wormald's machinery will address this issue. }} \] 

With this in mind, let us now consider informally what can be expected to happen in the neighbourhood of $l_i$. Suppose that $l_i$ is initially of type $\beta$ and is unhappy in the initial configuration. Then with  probability close to 1 for sufficiently large $w$, there will not be any unhappy nodes of type $\alpha$ in the neighbourhood of $l_i$ in the initial configuration. If $l_i$ changes type, then this will make the bias in its neighbourhood still more positive, which may cause further nodes of type $\beta$ to become unhappy. If these change to type $\alpha$ then this will further increase the bias, potentially causing more nodes to become unhappy, and so on. The following definitions formalise some of the  ways in which this process might play out, and in particular the possibility that this process might play out without interference from what happens in other neighbourhoods $\mathcal{N}(l_j)$ or $\mathcal{N}(r_j)$. 

\begin{defin} 
For $0<i<k_0$ we say that $l_i$  \textbf{completes at stage} $s$ if both: 

\begin{enumerate} 
\item No node in $\mathcal{N}(l_i)$ is unhappy at stage $s$, and this is not true for any $s'<s$. 
\item There exist $x_0$ and $x_1$ with $l_{i+1}+2w<x_0 < l_i -2w<l_i+2w<x_1< l_{i-1}-2w$, such that by the end of stage $s$, no node in $[x_0-w,x_0]$ or $[x_1,x_1+w]$ has changed type. 
\end{enumerate} 
We say that $l_i$ \textbf{completes} if it completes at some stage.  We also define completion for $r_i$ analogously.  
\end{defin} 

\begin{defin} 
We say that $l_i$ (or $r_i$) \textbf{originates a firewall} if it completes at some stage $s$ and (i) it belongs to firewall at stage $s$, and (ii) all type changes in $\mathcal{N}(l_i)$ (or $\mathcal{N}(r_i)$) at stages $\leq s$ are of the same kind (i.e. all $\alpha$ to $\beta$, or all $\beta$ to $\alpha$).\footnote{This definition might initially seem to neglect the possibility, for example, that $l_i$ completes at some stage and does not belong to a firewall at that stage, but that, nevertheless, the sequence of type changes in its neighbourhood and surrounding neighbourhoods which have led to completion have caused the creation of a firewall. In fact we shall be able to ensure (Lemma \ref{dichotomy}) that with probability close to 1 (for large $w$) such events do not occur.} 
\end{defin}

The informal idea, is that we now wish to show that each $l_i$ and each $r_i$ has some reasonable chance of originating a firewall (and that this reasonable chance is bounded below by some value which doesn't depend on $w$). Then we can choose $k_0$ so that the probability none of the $l_i$ originate a firewall or none of the $r_i$ originate a firewall is $\ll \epsilon $, i.e.\ with  probability close to 1 firewalls will originate either side of $u_0$ within the interval $[l_{k_0}, r_{k_0}]$. Then, letting $i_1$ be the least $i$ such that $l_i$ originates a firewall, and letting $i_2$ be the least $i$ such that $r_i$ originates a firewall, we wish to show that with probability close to 1 the firewalls originated at $l_{i_1}$ and $r_{i_2}$  will spread until $u_0$ is contained in one of them. Since these two firewalls have originated at nodes which are at distance at least $e^{w/d}$ apart, $u_0$ ultimately belongs to a firewall of at least this length. So to sum up:  

\begin{quote} 
The approximate reason Theorem \ref{3} holds is that $u_0$ can be expected to join a firewall which --  \emph{precisely because} unhappy nodes are rare in the initial 
configuration --  originated at a long distance from $u_0$. 
\end{quote}      

In order to make this basic picture work, however, we need to be careful about the formation of stable intervals in $[l_{k_0}, r_{k_0}]$. As noted above,  firewalls will spread until they hit stable intervals of the opposite type. Now suppose that, with $i_1$ and $i_2$ as above,  $i_1=i_2=2$ and, for now, suppose that $\alpha$-firewalls are originated at both $l_2$ and $r_2$.  In order to show that these two  firewalls will spread until they meet each other, it will be helpful first of all, to be able to assume that in the initial configuration there are no stable subintervals of $[l_{k_0}, r_{k_0}]$. This will follow quite easily for large $w$, from our previous analysis of the ratio between  the probability of unhappy nodes and stable intervals. A further danger that we have to be able to avoid, however, is that, while $l_1$ and $r_1$ do not originate firewalls, they \emph{do} get as far as creating $\beta$-stable intervals.

\begin{defin} Given $i$ with $0<i<k_0$, let $u=l_i$ or $u=r_i$. Let $u_1=u-(2w+1)$ and $u_2=u+(2w+1)$. 
We say that $u$ \textbf{subsides} if it completes at some stage $s$, and:
\begin{itemize} 
\item  There are no nodes in $[u_1,u_2]$ belonging to stable intervals at stage $s$; 
\item No nodes in $\mathcal{N}(u_1)$ or $\mathcal{N}(u_2)$ have been unhappy at any stage $\leq s$. 
\end{itemize} 
\end{defin}

 So we need to be able to show, in fact, that with probability close to 1 each $l_i$ and $r_i$ either originates a firewall or subsides.  To do this clearly involves a careful analysis what is likely to happen in each of the neighbourhoods $\mathcal{N}(l_i)$ and $\mathcal{N}(r_i)$.  First of all, the  large distances between these nodes mean that, for fixed $k_0$ and sufficiently large $w$, we can expect all of the $l_i$ and $r_i$ (for $0<i<k_0$) to complete, so that one can understand the early stages of the process for each of these neighbourhoods by considering each in isolation. We then wish to show: 

\begin{quote} The required  dichotomy: in the neighbourhood of each $l_i$ and $r_i$, either a small number of type changes will occur before completion,  or else a large number of type changes will occur before completion and a firewall will be created.
\end{quote} 

 Now, if we strengthen our original requirement that there are no stable subintervals of $[l_{k_0}, r_{k_0}]$ in the initial configuration, to a requirement that there are no subintervals which are `close' to being stable in the initial configuration (where `close' is to be made precise in such a way as to ensure that when a small number of type changes occur in the neighbourhood of $l_i$  before completion, these are not enough to create any stable intervals), then we shall have that with probability close to 1 each $l_i$ and each $r_i$ either originates a firewall or else completes without creating stable intervals. \\

Once all this is in place, there is then one further hurdle. In the above, we assumed that the firewalls originating at $l_2$ and $r_2$ are both $\alpha$-firewalls. If they are firewalls of opposite type, however, we still have some work to do in order to prove that $u_0$ will almost certainly end up belonging to one of these two firewalls.   \\

 \begin{figure} \label{first} 
\epsfig{file=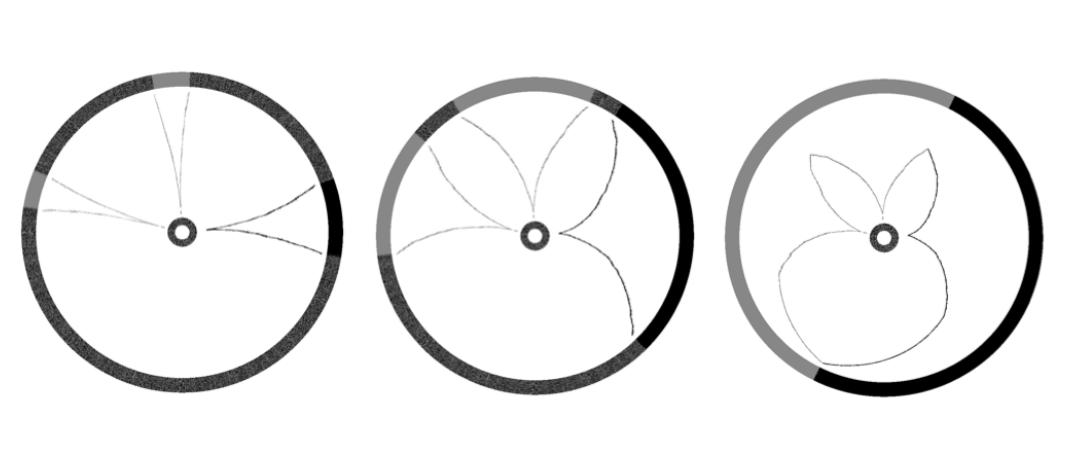,width=14cm}
 \caption{An example of the developing process with $n=100000, w=75, \tau=0.37$, at stages $s= 5000, 15000$ and $24986$.} 
\label{fig:process}
\end{figure}

 \textbf{Formalising the intuitive picture.}  Thus far we have defined an event, $\mbox{\textbf{Good spacing}}$, which depends upon  the value $k_0$. We are yet to specify $k_0$, but have promised that this choice will not depend on $w$.  For any $\epsilon'>0$ and for any fixed $k_0\geq 1$  we can ensure that $\textbf{P}(\mbox{\textbf{Good spacing}})>1-\epsilon'$ by taking $w$ sufficiently large (and by taking $n$ sufficiently large compared to $w$). 

In order to specify how the type changing process can be expected to develop in the vicinity  of $u_0$, we shall now proceed to define a finite number of other events of this kind. This finite set of events (having $\mbox{\textbf{Good spacing}}$ as a member)  we call shall call $\Pi$.   Our aim is to show that, for any $\epsilon'>0$, the probability that all the  events in $\Pi$ occur is greater than $1-\epsilon'$ for all sufficiently large $w$ (and for $n$ sufficiently large compared to $w$). Suppose that  this is established for some $\Pi' \subset \Pi$.    To establish the result  for $\Pi' \cup \{ Q \}$ with $Q\in \Pi-\Pi'$,  it then suffices to prove for each $\epsilon'>0$ and all sufficiently large $w$, that the probability of $Q$ given $P$ is greater than $1-\epsilon'$, where $P$ is any conjunction (possibly empty) of the events in $\Pi'$.  Of course, we choose that $P$ which is most convenient to work with. 

In discussing  the `required dichotomy'  above,  a requirement was suggested, that there should be no subintervals of $[l_{k_0}, r_{k_0}]$ which are `close' to being stable in the initial configuration. In fact an appropriate formalisation of this idea is easy to describe, and we now do so.  

\begin{defin} 
For any $\tau'\in (0,1)$, we say that an interval of length $w+1$ is $\tau'$\textbf{-stable}, if it contains $\tau'(2w+1)$ many $\alpha$-nodes or $\tau'(2w+1)$ many $\beta$ nodes.
\end{defin} 

\begin{defin}[$\mbox{\textbf{Stable clear}}$]  \label{tau0}
Once and for all, fix some $\tau_0$ with $\kappa <\tau_0 <\tau$. Let $\mbox{\textbf{Stable clear}}\in \Pi$ be the event that $l_i$ and $r_i$ are  defined for all $1\leq i \leq k_0$, and there do not exist any $\tau_0$-stable subintervals of $[l_{k_0},r_{k_0}]$ in the initial configuration.  
\end{defin}

The following lemma has a simple proof, given in Section \ref{def2}: 

\begin{lem} \label{stableclear}  
For fixed $k_0$ and $\epsilon'>0$,  $\textbf{P}(\mbox{\textbf{Stable clear}})>1-\epsilon'$ for all  $w$ sufficiently large (and all $n$ sufficiently large compared to $w$). 
\end{lem} 

 Now, in order to establish the required dichotomy, we need to build up a clear picture of what the neighbourhoods $\mathcal{N}(l_i)$ and $\mathcal{N}(r_i)$ can be expected to look like. The distributions for these intervals in the initial configuration are a little difficult to attack directly, however,  due to the nature of their definition. In choosing $l_1$ we move left until we find the \emph{first} node which has high bias -- this gives an asymmetry to the given information concerning the neighbourhood. Roughly, we might expect something like a hypergeometric distribution, but how good is this as an approximation? What we shall do, in fact, is first of all to understand what can be expected from the neighbourhood of a node which is chosen uniformly at random from among those with borderline bias:

\begin{defin}
Let us say that a node $u$ has \textbf{borderline bias}, denoted $\mathtt{Bb}(u)$, if: 

\[ |\Theta(\mathcal{N}(u))|=\mbox{Min}\{\theta \in 2\mathbb{N}+1: \theta > (1 - 2\tau)(2w+1)\}, \] 

 i.e.,  $u$ has high bias but  decreasing the modulus of the bias by the minimum possible amount of 2 would cause it not to have high bias. 
We say that $\mathtt{Bb}^{\ast}(u)$ holds if $u$ has borderline bias in the initial configuration. 
\end{defin}

 The following is worth emphasising: 
\begin{enumerate} 
\item[$(\dagger_1)$]  Each of the nodes $l_i$ and $r_i$ ($0<i\leq k_0$) has borderline bias.  
\end{enumerate}

 In what follows it is often convenient to work with some fixed $k\geq 1$ and to divide an interval $I=[a,b]$ into $k$ parts of equal length. This occasions the minor inconvenience that the length of the interval might not be a multiple of $k$, motivating the following definition: 

\begin{defin} \label{divide} Let $I=[a,b]$ and suppose $k\geq 1$.  We define the subintervals 
$I(1:k):= \left[ a, a+\left\lfloor \frac{b-a}{k} \right\rfloor \right]:=$ $\left[ I(1:k)_1,I(1:k)_2 \right]$ and $$I(j:k) := \left[ a+ \left\lfloor \frac{(j-1) (b-a)}{k} \right\rfloor +1, a+\left\lfloor \frac{j (b-a)}{k} \right\rfloor \right] := \left[ I(j:k)_1,I(j:k)_2 \right]$$ for $2 \leq j \leq k$.

 \end{defin} 
 
 In Definition \ref{divide} the intervals are counted from left to right, but it is also useful sometimes to work from right to left: 
 
 \begin{defin} Let $I=[a,b]$ and suppose $k\geq 1$. For $1\leq j\leq k$ we define $I(j:k)^-= I(k-j+1:k)$,  $I(j:k)^-_1= I(k-j+1:k)_1$ and $I(j:k)^-_2=I(k-j+1:k)_2$. 
  \end{defin}

\begin{lem}[Smoothness Lemma]  \label{smooth} Suppose $u$ is such that the proportion of $\alpha$ nodes in $I:=\mathcal{N}(u)$ is $\theta$, and that $u$ is selected uniformly at random from nodes with this property. Then for any fixed $k\geq 1$ and $\epsilon'>0$, for all sufficiently large $w$ the following holds with probability $>1-\epsilon'$: for every $j$ with $1\leq j \leq k$, the proportion of the nodes in $ I(j:k)$ which are $\alpha$ nodes,  lies in the interval  $[\theta-\epsilon', \theta+\epsilon']$. 
 \end{lem}

 \begin{proof} 
 Once we are given that $\theta$ is the proportion of the nodes in $I$ which are of type $ \alpha$, the nodes in this interval cease to be independently distributed, and the distribution becomes hypergeometric. Since we consider fixed $k$ and $\epsilon'$ and take $w$ large, it suffices to prove the result for given $j$ with $1\leq j \leq k$, i.e.\ if $P_1,...,P_k$ are events, each of which occurs with probability tending to 1 as $w\rightarrow \infty$, then their conjunction also occurs with probability tending to 1 as $w\rightarrow \infty$.  For a given $j$, the result follows directly, however, from an application of Chebyshev's inequality and standard results for the mean and variance of a hypergeometric distribution. Let $ x$ be the number of $\alpha$ nodes in the interval $I(j:k)$ and let $\ell$ be the length of the interval, so that $|\ell -(2w+1)/k|\leq 1$. Then we have: 
 \[ \textbf{P}(|x/\ell-\theta|>\epsilon')<\ell^{-2} \epsilon'^{-2} \mbox{ Var}(x) =O(1) \ell^{-1}. \] 
  \end{proof} 

Lemma \ref{smooth} basically tells us that if we choose a node $u$ with borderline bias uniformly at random, then for large $w$ we can expect the bias to move towards 0 fairly smoothly as we move to $u+(2w+1)$ or $u-(2w+1)$. In order to see roughly why this is true, suppose that  $|\Theta(\mathcal{N}(u))|=\rho$ and let $\theta$ be the proportion of the nodes in $\mathcal{N}(u)$ which are of type $\alpha$. Let  $I=[u,u+(2w+1)]$ and, for some $k$, consider the sequence of evenly spaced nodes $v_j= I(j,k)_2$.  Now in forming the neighbourhood of $v_j$, we lose an interval of length (almost exactly) $\lfloor j(2w+1)/k \rfloor$ from the neighbourhood of $u$, which  by Lemma \ref{smooth} we can expect to have a proportion of $\alpha$ nodes very close to $\theta$. We also gain an interval of the same length from outside $\mathcal{N}(u)$, which we can expect to have a proportion of $\alpha$ nodes very close to $\frac{1}{2}$.  This means a bias for $v_j$  close to $\rho \frac{k-j}{k} $. 

The following definition allows us to express this more formally:

\begin{defin} Suppose that $\mathtt{Hb}^{\ast}(u)$ holds. 
Let $I_1= [u-(2w+1),u]$ and $I_2=[u,u+(2w+1)]$. Let $|\Theta(\mathcal{N}(u))|=\rho$ and let $\theta$ be the proportion of the nodes in $\mathcal{N}(u)$ which are of type $\alpha$. Suppose that $k\geq 1$ is even and $\epsilon'>0$. For $1\leq j \leq k$ let $v_j=I_2(j:k)_2$ and let  $v_{-j}=I_1(j:k)^-_1$. We say that  $\mathtt{Smooth}_{k,\epsilon'}(u)$ holds if  both: 

\begin{itemize} 
\item For every $j$ with $1\leq |j| \leq k$,  $|\Theta(\mathcal{N}(v_j))-\rho \frac{k-j}{k}|/w < \epsilon'$. 
\item For every $j$ with $1\leq j\leq k/2$ the proportion of the nodes in $I_1(j:k)^-$ which are of type $\alpha$ lies in the interval $[\theta-\epsilon',\theta+\epsilon']$, and similarly for $I_2(j,k)$. 
\end{itemize} 

 We say that $\mathtt{Smooth}_{k,\epsilon'}^{\ast}(u)$ holds if $\mathtt{Smooth}_{k,\epsilon'}(u)$ holds in the initial configuration. Figure \ref{fig:definingthel_2} illustrates the picture for $k=2$. 
\end{defin}

\begin{figure}
\epsfig{figure=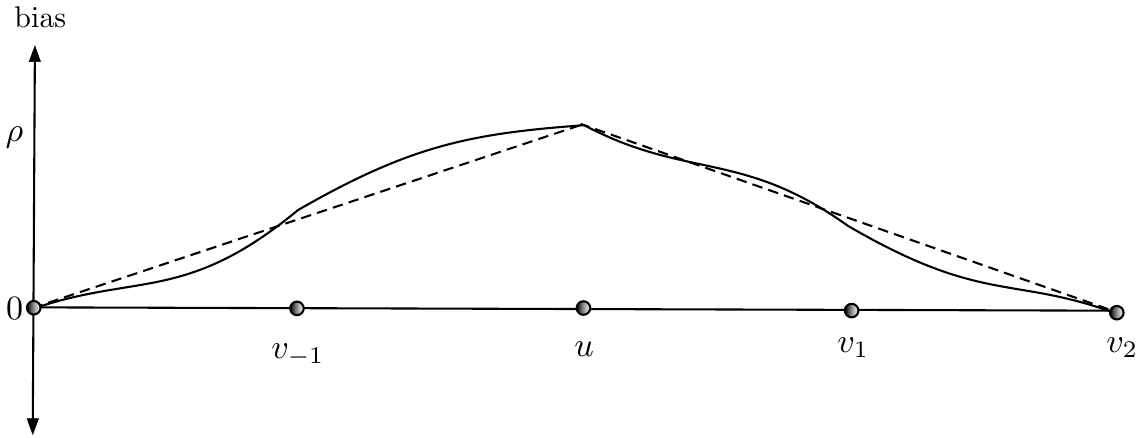,width=11cm}
\caption{Smooth bias change for $k=2$.}
\label{fig:definingthel_2}
\end{figure}

\begin{coro}[Smoothness Corollary] \label{smooth2} Suppose $u$ is selected uniformly at random among those nodes such that $\mathtt{Bb}^{\ast}(u)$ holds.  For all $k\geq 1$ and $\epsilon'>0$, and for all sufficiently large $w$, $\mathtt{Smooth}^{\ast}_{k,\epsilon'}(u)$ holds with probability $>1-\epsilon'$. 
\end{coro} 
\begin{proof}  Let $u_{-1}=u-(2w+1)$ and $u_1=u+(2w+1)$. The fact that $u$ has borderline bias has no impact on the distributions for $\mathcal{N}(u_{-1})$ and $\mathcal{N}(u_1)$. Let $I_3=\mathcal{N}(u_{-1})$ and $I_4=\mathcal{N}(u_{1})$. By the weak law of large numbers, for large $w$ we can expect the proportion of $\alpha$ nodes  in each  $I_3(j:k)$ and $I_4(j:k)$ to be close to $\frac{1}{2}$. The result then follows from Lemma \ref{smooth}.  
\end{proof}

While  Smoothness Corollary \ref{smooth2} tells us what can be expected from the neighbourhood of a node chosen uniformly at random from among those with borderline bias, this does not immediately allow us to infer anything about what can be expected from the neighbourhood of each $l_i$ and $r_i$. What we need is that, if we choose a node $u$ uniformly at random and then move left (or right) until we find a first node $v$ with high bias, then with probability close to 1 $\mathtt{Smooth}_{k,\epsilon'}^{\ast}(v)$ holds.  Ideally, we would like to be able to apply Lemma \ref{gen}, but that would involve establishing an appropriate form of condition (ii) from the statement of the lemma. In order to work around this condition, one has to take a slightly circuitous route via several intermediate notions, but the following lemma can be established (the proof is given in Section \ref{def2}). 

\begin{lem}[Smoothness for $l_i$ and $r_i$]  \label{lismooth}
For any node $u$, let $x_u$ be the first node to the left of $u$ which has high bias in the initial configuration.  For any $\epsilon'>0$ and $k\geq 1$, if $0 \ll w \ll n$ and $u$ is chosen uniformly at random, then  $x_u$ is defined and  $\mathtt{Smooth}^{\ast}_{k,\epsilon'}(x_u)$ holds with probability $>1-\epsilon'$. 

An analogous result holds when `left' is replaced by `right'. 
\end{lem}

In order to see that Lemma \ref{lismooth} suffices to establish probable smoothness for all of the $l_i$ and $r_i$, note first that $k_0$ is fixed while we take $w$ large. At step $i$ of the iteration which defines the sequence $l_1,l_2,..$, the fact that $l_i$ has borderline bias tells us nothing about the neighbourhood of $l_i-(2w+1)$ or the nodes to the left of this neighbourhood (but at a distance small compared to $n$).  

We are now ready to define the third event in $\Pi$:

\begin{defin}[$\mbox{\textbf{Smooth}}$]  \label{k1k2} Let $\tau_0$ be as in Definition \ref{tau0}. Once and for all, choose $k_1$ such that $\frac{1}{k_1}\ll \tau-\tau_0$, and choose $k_2$ and $\epsilon_0$ such that $k_1 \ll k_2 \ll \frac{1}{\epsilon_0}$ and $k_2$ is a multiple of $k_1$.  We define $\mbox{\textbf{Smooth}}$ to be the event that all the $l_i$ and $r_i$ are defined for $1\leq i \leq k_0$, and that when $u=l_i$ or $u=r_i$, $\mathtt{Smooth}_{k_2,\epsilon_0}^{\ast}(u)$ holds. 
\end{defin} 

By Lemma \ref{lismooth}, when $k_0 \ll w$ the probability that $\mbox{\textbf{Smooth}}$ does not occur is $\ll \epsilon$.

\subsection*{The process to completion} Having established a clearer picture of what can be expected from the initial configuration, we now look to understand what will happen in the early stages, in the neighbourhood of each $l_i$ or $r_i$. First of all, we must establish that these nodes can be expected to complete.  The proof of Lemma \ref{complet} is given in Section \ref{def2}, but the reader can probably see immediately how the proof might go.  A happy node cannot become unhappy without changes to their neighbourhood -- meaning that unhappiness has to spread one neighbourhood at a time. Given the large distances between the $l_i$ and $r_i$, failure of completion for some $l_i$ would involve a large number of events causing the spread of unhappiness from $l_{i-1}$ or $l_{i+1}$, all occurring during stages when there are still unhappy nodes in $\mathcal{N}(l_i)$. At each stage the likelihood of such an event is not that much more than an unhappy node in $\mathcal{N}(l_i)$ being chosen to swap.     

\begin{lem}[$l_i$ and $r_i$ complete]  \label{complet} 
For any $\epsilon'>0$, if $0\ll w \ll n$ then for  all $i\in [1,k_0)$, $l_i$ and $r_i$ will (be defined and will) complete with probability $>1-\epsilon'$. 
\end{lem}

\begin{defin}[$\mbox{\textbf{Completion}}$]  
We define $\mbox{\textbf{Completion}}$ $(\in \Pi)$ to be the event that all the $l_i$ and $r_i$  (for $1\leq i <k_0$) are defined and complete. 
\end{defin}

We are now ready to prove the required dichotomy. 

\begin{lem}[The required dichotomy]  \label{dichotomy} Suppose that $w$ is large and that $\mbox{\textbf{Good spacing}}$, $\mbox{\textbf{Stable clear}}$, $\mbox{\textbf{Smooth}}$  and $\mbox{\textbf{Completion}}$ all hold. 
Then for $i<k_0$,  $l_i$  and $r_i$  will each either subside or originate a firewall. 
\end{lem}  
\begin{proof} We prove the result for $l_i$, and the proof for $r_i$ is essentially identical.

 Note first that the choice of $k_1$ in Definition \ref{k1k2} means, in particular, that $10 w/k_1$ type changes in any given neighbourhood cannot create stable intervals, given that $\mbox{\textbf{Stable clear}}$ holds (the numbers here are fairly arbitrary).  Note also, that satisfaction of $\mbox{\textbf{Smooth}}$ suffices to ensure that there are not unhappy nodes of both types in the neighbourhood of $l_i$ in the initial configuration. 
Now suppose that $l_i$ completes at stage $s$ and has positive bias in the initial configuration (the case for negative bias is essentially identical).  

It is useful at this point to establish names for a number of relevant intervals. We let $u_1=l_i-(2w+1)$ and $u_2=l_i +(2w+1)$. Then we define: 

\begin{itemize} 
\item $J=\mathcal{N}(u_1) \cup \mathcal{N}(l_i) \cup \mathcal{N}(u_2)$. 
\item $I=[u_1,u_2]$. 
\item $I_1=[u_1,l_i]$, $I_2=[l_i,u_2]$. 
\item $K_j^1=I_1(j:k_1)^- \cup I_2(j:k_1)$. 
\item $K_j^2=I_1(j:k_2)^- \cup I_2(j:k_2)$. 
\end{itemize}

In Definition \ref{k1k2} we assumed that $k_2$ is a multiple of $k_1$, so  we may let $m$ be such that $k_2=mk_1$. It is convenient to assume that $k_2$ is even.  Now we divide into two cases. 

$l_i$ \textbf{subsides}.  First of all, suppose that at stage $s$, there is a $\beta$-node $u$  in the interval $K_1^1$. Then $u$ must be happy at stage $s$. The fact that $\mbox{\textbf{Smooth}}$ is satisfied, together with the fact that $l_i$ completes at stage $s$, means that prior to stage $s$, the only type changes in the interval $J$ are from type $\beta$ to type $\alpha$, so that $u$ must be happy at \emph{every} stage $\leq s$. Now, since $k_1 \ll k_2 \ll \frac{1}{\epsilon_0}$ and   $\mathtt{Smooth}_{k_2,\epsilon_0}^{\ast}(l_i)$ holds, any nodes in $I- (K_1^1 \cup K_2^1)$ have lower bias than $u$, and hence are happy,  in the initial configuration. It then follows  by induction on the stages $\leq s $ that no node in $J- (K_1^1 \cup K_2^1)$ changes type prior to stage $s$.  In order to see this suppose that it holds prior to stage $s'\leq s$. Then at stage $s'$, if $v\in I- (K_1^1 \cup K_2^1)$, it still has lower bias than $u$ and so cannot change type from $\beta$ to $\alpha$, since $u$ is happy so $v$ must be. If $v\in J-I$, then $v$ changing type would contradict the fact that $l_i$ completes. 

  We therefore get at most  $|(K_1^1 \cup K_2^1)|<10w/k_1$ many type changes in the interval $J$ prior to completion. As observed above, this means that no stable intervals are created and $l_i$ subsides, as required.

$l_i$ \textbf{originates a firewall}. So suppose instead that, at stage $s$, all nodes in the interval $K_1^1$ are of type $\alpha$.  Given $m$ as above, another way of putting this, is that all nodes in $\bigcup_{j\leq m} K^2_j$ are of type $\alpha$ at stage $s$. We now show by induction on $r\geq m$ that, when $r\leq k_2/2$, any  nodes in $K_r^2$ must be of type $\alpha$  at stage $s$ -- i.e. that all nodes in $\mathcal{N}(l_i)$ are $\alpha$ nodes at stage $s$. So suppose that  $m\leq r<k_2/2$ and that the hypothesis holds for all $r'\leq r$. Consider $u \in K_{r+1}^2$. Let $\rho$ be the bias of $l_i$ in the initial configuration. First let us form a lower bound for the bias of $u$ in the initial configuration. The fact that $\mbox{\textbf{Smooth}}$ holds means that the leftmost and rightmost nodes in $K_r^2$ have bias at least $\rho-\frac{r}{k_2}\rho -\epsilon_0 w$. Then, since the bias can change by at most 2 if we move left or right one node, we conclude that in the initial configuration $u$ has bias 
\[ \rho_1 \geq  \rho - \left( \frac{r}{k_2}\rho + \epsilon_0 w +\frac{2(2w+1)}{k_2}\right). \]

Now we have to take into account all of the $\beta$ nodes in $\bigcup_{j\leq r} K^2_r$ which have changed type. In fact, so that we can be sure that each change of type affects the bias of $u$, we shall consider just those which lie between $l_i$ and $u$. Let $\theta$ be the proportion of the nodes in $\mathcal{N}(l_i)$ which are $\alpha$ nodes in the initial configuration (recalling that $l_i$ has borderline bias at that point) -- so that $\rho =(2\theta -1)(2w+1)$. Then the number of $\beta$ nodes in $\bigcup_{j\leq r} K^2_j$ in the initial configuration, which lie between $l_i$ and $u$, is at least: 
\[ (1-\theta-\epsilon_0)(2w+1)r/k_2. \] 

 Each change of type for one of these nodes means an increase of 2 in the bias of $u$, so that at stage $s$, $u$ has bias:

\begin{align*}
\rho_2    &  \geq  \rho  - \left( \frac{r}{k_2}\rho + \epsilon_0 w +\frac{2(2w+1)}{k_2}\right) + (1-\theta-\epsilon_0)2(2w+1)r/k_2. \\
              & = \rho + (2w+1) \left(   \frac{ (1-\theta)2r}{k_2}    -\frac{\epsilon_0 w}{2w+1} -   \frac{2}{k_2}                -     \frac{r}{k_2}(2\theta -1)   -  \frac{\epsilon_0 2r}{k_2} \right)           \\ 
              \end{align*}

 We are left to compare the terms \[  \frac{ (1-\theta)2r}{k_2}, \ \ \frac{\epsilon_0 w}{2w+1}, \ \  \frac{2}{k_2}, \ \ \frac{r}{k_2}(2\theta -1)  \mbox{ and } \frac{\epsilon_0 2r}{k_2}. \]

 Since $1/k_2\gg \epsilon_0$, the second term is much smaller than the third. Since $\theta \in (0.5,0.65)$, $r/k_2\geq 1/k_1$ and $k_1\ll k_2$, the first term is much larger than the third (to see that $\theta \in (0.5,0.65)$ recall that $l_i$ has borderline bias in the initial configuration and $\tau>\kappa>0.35$). Since $\epsilon_0$ is small, the first term is also much larger than the last. The result then follows for large $w$, since $2-2\theta$ is always more than double $2\theta -1$ for  $\theta\in (0.5,0.65)$, meaning that the first term is more than double the fourth term, and thus $\rho_2 \geq \rho$, meaning that if $u$ is a $\beta$ node, it will be unhappy.
\end{proof}

 \begin{figure} \label{firewall} 
\includegraphics[scale=0.6, clip =true, trim=2cm 0.5cm 1cm 0.5cm]{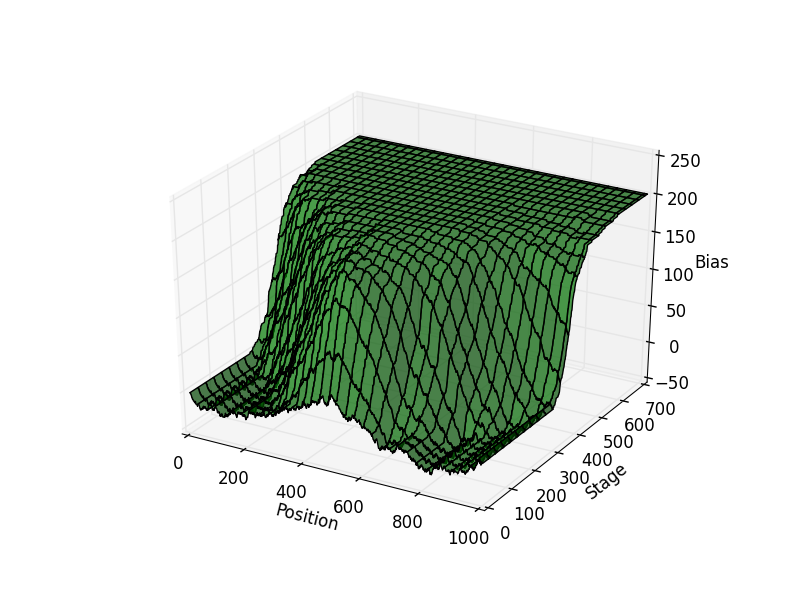}
 \caption{The creation of a firewall (extracted from simulation).}
\label{fig:firewall}
\end{figure}

\begin{lem}[Reasonable chance of firewall] \label{fork0} 
Suppose that $\mbox{\textbf{Good spacing}}$ holds. There exists $\delta>0$ which does not depend on $w$, such that if $1\leq i <k_0$, then $l_i$ originates a firewall with probability $>\delta$ (and similarly for $r_i$). 
\end{lem} 
\begin{proof} 
Rather than working directly with $l_i$, to begin with we consider $u$, which is chosen uniformly at random from amongst the nodes with borderline bias in the initial configuration. We extend the definitions of completion and firewall origination to apply to $u$ in the obvious way,  and we establish that there exists $\delta>0$ which does not depend on $w$, such that $u$  originates a firewall with probability $>\delta$. 
Let $\theta$ be the proportion of the nodes in $\mathcal{N}(u)$ which are of type $\alpha$ in the initial configuration.   Without loss of generality, suppose that $u$ has positive bias $\rho$.  Let $K_1^1(u)$ be as defined in the proof of Lemma \ref{dichotomy}, but relative to $u$, i.e.\ with $u$ replacing $l_i$ in that definition.  We first look to establish first that there exists $\delta'>0$ which does not depend on $w$, such that the probability of $u$ completing with all nodes in $K^1_1(u)$ being of type $\alpha$, is greater than $\delta'$. 

 Initially let us adopt the approximation that the nodes in $\mathcal{N}(u)$ are i.i.d.\ with probability $\theta$ of being an $\alpha$ node.

We consider a process, consisting of $\ell :=\lceil \frac{2w+1}{k_1} \rceil +1$ many steps. 

\textbf{Step} $0$. If $u$ is of type $\beta$ then define $\rho_0=\rho +2$, and otherwise define $\rho_0=\rho$. 

\textbf{Step} $s>0$ (with $s\leq \ell$). Consider the two nodes $u -s$ and $u +s$. Let $x_s$ be the number of these which are $\beta$ nodes (so $x_s\in \{ 0,1,2 \}$).  Let $y_s$ be $\Theta^{\ast}(\mathcal{N}(u-s)) -\Theta^{\ast}(\mathcal{N}(u-s+1))$. Then define $\rho_s= \rho_{s-1} + 2x_s +y_s$.   

So, given the approximation that the nodes in $\mathcal{N}(u)$ are i.i.d.\ with probability $\theta$ of being an $\alpha$ node, this gives a biased random walk.  We may also consider the mirror image process, in which $y_s$ is defined in terms of the bias at $u+s$ and $u+s-1$ instead of the bias at $u-s$ and $u-s+1$. Suppose that this process gives a set of values $\rho_s'$. If $\rho_s>\rho$ and $\rho'_s>\rho$ for all $s$ then let us say that $u$ is \emph{prone to firewall origination} - the idea is that if $u$ also completes and only has one type of unhappy node in $\mathcal{N}(u)$ in the initial configuration (though we do not yet assume these conditions), then being prone to firewall origination will mean that all the nodes in $K^1_1(u)$ must be $\alpha$ nodes at completion. To show that $u$ has some reasonable chance of being prone to firewall origination, it suffices (modulo our false assumption that the nodes in $\mathcal{N}(l_i)$ are i.i.d.) to show that $\rho_s-\rho_{s-1}$ is more likely to be positive than negative, for $0<s\leq \ell$. 

The probability that $\rho_s-\rho_{s-1}=-2$ is $\frac{1}{2}\theta^3$. The probability that $\rho_s-\rho_{s-1}=0$ is $\frac{1}{2}\theta\cdot 2\theta(1-\theta) + \frac{1}{2}\theta^2= \theta^2(1-\theta)+  \frac{1}{2}\theta^2$. 
Therefore the probability that  $\rho_s-\rho_{s-1}\in \{ 2,4,6 \} = 1- \frac{1}{2}\theta^3-\theta^2(1-\theta)-\frac{1}{2}\theta^2>\frac{1}{2}$, given that $l_i$ has positive borderline bias in the initial configuration, and so  $\theta \in (0.5,0.65)$. 

Now we must drop the false assumption of independence which was made earlier, and work with the hypergeometric distribution on $\mathcal{N}(u)$. In the random walk described above, let   $p>\frac{1}{2}$ be the probability that $\rho_s-\rho_{s-1}>0$. Now choose $p'$ with $\frac{1}{2}<p'<p$. For sufficiently large $k_1$,\footnote{This works for all $k_1>c$ for some constant $c$ (which can actually be chosen so as not to depend even on $\tau$, given that $\kappa<\tau<\frac{1}{2}$), so strictly speaking we should have required in Definition \ref{k1k2} that $k_1>c$. Since this would have been confusing at that point, we mention it now instead, and observe that there is no circularity, i.e.\ $c$ certainly need not depend on $k_1$ or any values defined in terms of $k_1$.}  when we drop the assumption of independence, the actual probability that $\rho_s-\rho_{s-1}>0$ at each step is greater than $p'$, no matter what has occurred at previous steps.

So far then, we have been able to conclude that there exists $\delta'>0$, which does not depend on $w$, such that $u$  has probability $>\delta'$ of being prone to firewall origination. Next, note that by Corollary  \ref{smooth2} and by the same argument as in the proof of Lemma  \ref{complet},  the probability that $u$ completes and satisfies  $\mathtt{Smooth}_{k_2,\epsilon_0}^{\ast}(u)$ tends to 1 as $w\rightarrow \infty$. Satisfaction of these conditions combined with being prone to firewall origination,  means that all the nodes in $K^1_1(u)$ must be $\alpha$ nodes at completion. Then, as observed in the proof of Lemma \ref{dichotomy}, this latter condition combined with satisfaction of $\mathtt{Smooth}_{k_2,\epsilon_0}^{\ast}(u)$ means that, in fact, all nodes in $\mathcal{N}(u)$ must be $\alpha$ nodes when $u$ completes. 

Finally, we move to consider $l_i$ rather than $u$. Given $l_i$, choose $v$ uniformly at random from amongst the nodes in $\mathcal{N}(l_i)$ with borderline bias. Now the information we have about $v$ completely specifies the probability that $v$ originates a firewall, and at this point, since we have only assumed $\mbox{\textbf{Good spacing}}$, the only information we have about $v$ is:  $(\ast)$ $v$ has borderline bias, belongs to a configuration which satisfies $\mbox{\textbf{Good spacing}}$ relative to \emph{some} node (which happened to be $u_0$ in this case), and for the node labelled $l_i$ within that configuration,  belongs to $\mathcal{N}(l_i)$.  For $u$ chosen uniformly at random from amongst those with borderline bias,  the probability that $u$ satisfies $(\ast)$ tends to 1 as $w\rightarrow \infty$. So if $u$ has probability $\delta'>0$ of originating a firewall, then for any $\delta<\delta'$, we may conclude that a node chosen uniformly at random from amongst those which satisfy $(\ast)$, has probability $>\delta$ of originating a firewall so long as $w$ is sufficiently large. 
Thus $v$ has  probability $>\delta$ of originating a firewall, so long as $w$ is sufficiently large. 
Now we assume that $\mbox{\textbf{Completion}}$ and $\mbox{\textbf{Smooth}}$ both hold. Then $v$ must belong to $K^1_1(l_i)$, meaning that  if  $v$ originates a firewall then $l_i$ belongs to that same firewall when $v$ completes, and so must also originate a firewall. 
\end{proof} 

The following definition gives our final event in $\Pi$ and also specifies $k_0$. 

\begin{defin}[Defining $\mbox{\textbf{Firewall}}$ and choosing $k_0$]  
We let  $\mbox{\textbf{Firewall}}$ be the event that one of the $l_i$ $i<k_0$ is defined and originates a firewall, and that the same holds for some $r_j$, $j<k_0$. According  to Lemma \ref{fork0}, given $\epsilon>0$ we can choose $k_0$ once and for all, which is large enough such that the probability  $\mbox{\textbf{Firewall}}$ does not occur is $\ll \epsilon$ for  $0\ll w\ll n$.  
\end{defin} 

 The following lemma completes the proof of Theorem \ref{3} for the simple model. The proof is given in Section \ref{def2}. 

\begin{lem}[$u_0$ ultimately joins a firewall] \label{final} Suppose that all events in $\Pi$ hold.
Let $i$ be the least such that $l_{i}$ is defined and originates a firewall, and let $j$ be least such that $r_{j}$ is defined and originates a firewall. For any $\epsilon'>0$, for all sufficiently large $w$, with probability $>1-\epsilon'$,  $u_0$ will eventually be contained in one of the two firewalls originated at $l_i$ and $r_j$.
\end{lem}

\subsection*{The standard model}  

The difficulty that arises when one moves to the standard model, is that when there are different numbers of unhappy $\alpha$ and $\beta$ nodes, it is no longer true that every unhappy node is equally likely to be chosen as part of a swapping pair. If there are more unhappy $\alpha$ nodes than unhappy $\beta$ nodes at a given stage, for example, then unhappy $\beta$ nodes belong to more unhappy pairs of opposite type than do their unhappy $\alpha$ counterparts, and so are more likely to be chosen as part of an unhappy pair. This potentially complicates our proof of Lemma \ref{complet}, that each $l_i$ and $r_i$ will almost certainly complete, for example, since it may make `steps towards completion' (as defined in that proof) less likely. The solution, just as in \cite{BK}, is to use technology developed by Wormald \cite{NW} in order to show that we can sufficiently accurately model the discrete process with a continuous one which is governed by a system of differential equations, and thereby demonstrate that  the number of unhappy nodes of each type actually remains very evenly balanced.  Our use of the Wormald machinery is very similar to the corresponding argument in \cite{BK}. We give the full proof since there are some minor differences, and to make the paper as self-contained as possible. We aim to prove the following fact, which  shows that our proof of Theorem \ref{3}  for the simple model suffices for the standard model as well (the proof of $(\diamond)$ appears in Section \ref{def2}): 

\begin{enumerate} 
\item[$(\diamond)$] Suppose $w, \tau$ (with $\kappa<\tau <\frac{1}{2})$ and $\epsilon>0$ are fixed (where  $\epsilon$ is the value we have fixed throughout this section, and which played a role in the definition of $k_0$).  Let $l_{k_0}$ and $r_{k_0}$ be as defined previously. For any $\epsilon'>0$, when  $n$ is sufficiently large the following holds with probability $>1-\epsilon'$: there exists a first stage at which there are no unhappy nodes in the interval $[l_{k_0},r_{k_0}]$ and at all stages up to this the total number of unhappy $\alpha$ nodes divided by the total number of unhappy $\beta$ nodes lies in the interval $[1-\epsilon', 1+\epsilon']$.
\end{enumerate}

\subsection*{Increasing \texorpdfstring{$\tau$}{tau} in the interval \texorpdfstring{$[\kappa,\frac{1}{2}]$}{[k,0.5]} decreases run-lengths} Let us fix $\epsilon>0$ and  $\tau_1$ and $\tau_2$ such that $\kappa<\tau_1<\tau_2<\frac{1}{2}$.  For some large $w$ and $n\gg w$ suppose that we run two version of the process, one for $\tau_1$ and the other for $\tau_2$, and then let $\ell_1$ and $\ell_2$ be the corresponding run-lengths to which $u_0$ (chosen uniformly at random) belongs in the final configuration. Our aim in this subsection is to observe that, so long as $w$ is sufficiently large, the probability that $\ell_1>\ell_2$ is greater than $1-\epsilon$.   

In order to see this we may reason as follows. In the proof of  Theorem \ref{3}, we chose $k_0$ in such a way we could be almost certain at least one of the $l_i$ and at least one of the $r_i$ would originate firewalls. We could equally well have chosen $k_0$ so that the following fails to occur with probability $\ll \epsilon$: two of the $l_i$ originate firewalls of opposite type, and similarly two of the $r_i$ originate firewalls of opposite type.  Our previous analysis then suffices to show that the following fails to occur with probability $\ll \epsilon$ for sufficiently large $w$; $u_0$ ultimately belongs to a firewall of length $>\mbox{min}\{ |u_0-l_1|, |u_0 -r_1| \}$ and of length $<[l_{k_0},r_{k_0}]$.    So it suffices to show that with probability $>1-\epsilon$,  the length of the interval $[l_{k_0},r_{k_0}]$ for the process with $\tau_2$ is less than the value $\mbox{min}\{ |u_0-l_1|, |u_0 -r_1| \}$ for the process with $\tau_1$. This then follows for sufficiently large $w$, by applying  Lemma \ref{gen} to the events $P_u$, that $u$ is $\tau_2$-unhappy, and $Q_u$ that $u$ is $\tau_1$-unhappy.

\section{The case \texorpdfstring{$\tau=\frac{1}{2}$}{tau equals 1/2}} \label{half} This case was dealt with in \cite{BK}, for both the simple and standard models. 
Although we have talked about threshold behaviour occurring at both $\tau=\kappa$ and $\tau=\frac{1}{2}$, it is worth noting that these two thresholds are of essentially different types. For a fixed large $w$, as one increases the value of $\tau$ past the lower threshold $\kappa$, one will very suddenly observe entirely different results for the final configuration -- from a situation in which almost no type changing occurs, one moves to a situation where $u$ chosen uniformly at random can be expected to belong to a firewall of length exponential in $w$ in the final configuration. As one increases $\tau$ gradually up to $\frac{1}{2}$, however, what one observes is a smooth decrease in the expected length of the firewall to which $u$ belongs, with the situation only reversing very suddenly when $\tau>\frac{1}{2}$. In order to observe the difference in behaviour for when $\tau$ is just less than $\frac{1}{2}$ or when $\tau=\frac{1}{2}$, one should consider instead these two fixed $\tau$, and gradually increase $w$. Then, when $\tau<\frac{1}{2}$, the expected length of the firewall to which $u$ belongs grows exponentially, while for $\tau=\frac{1}{2}$ it grows only polynomially in $w$. Figure \ref{fig:st}  displays this difference in behaviour, and also illustrates the distinct mechanisms involved for the two cases. These diagrams illustrate the process exactly as in Figure \ref{fig:thresbehav}, the difference being that now $n=1000000$ rather than $100000$. It is also interesting to observe the difference in the way that the firewall spreads as a `wave' emanating from unhappy nodes in the initial configuration for different $\tau<\frac{1}{2}$. When $\tau$ is just above $\kappa$ the wave is sharply defined, (as illustrated in Figure \ref{fig:thresbehav}), much less so  when $\tau$ is just below $\frac{1}{2}$.

 \begin{figure}
\begin{center}
\epsfig{file=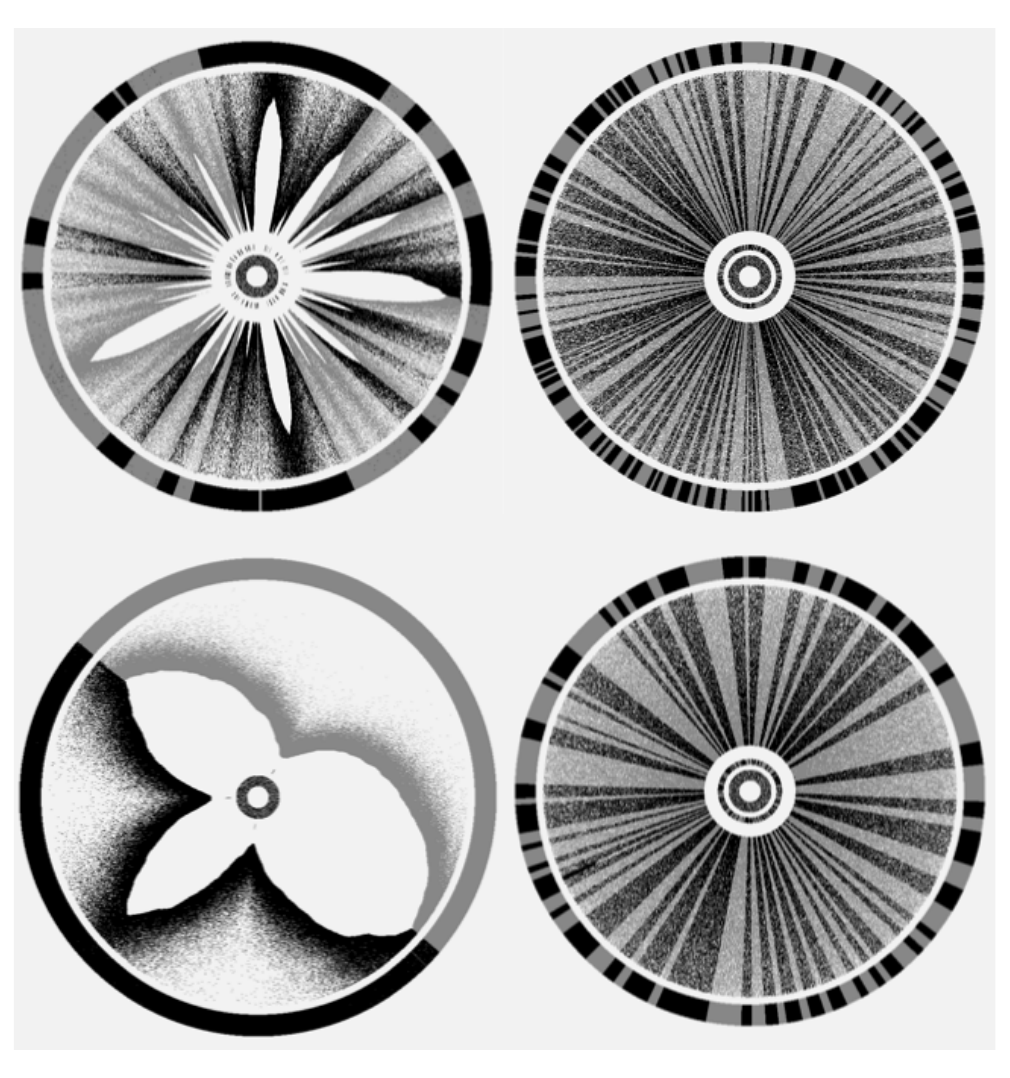,width=12cm}
\end{center}  
 \caption{The second kind of threshold behaviour: top left $w=1500, \tau=0.48$; bottom left $w=3000, \tau=0.48$; top right $w=1500, \tau=0.5$; bottom right $w=3000, \tau=0.5$.}
 \label{fig:st}
\end{figure}

\section{The case \texorpdfstring{$\tau>\frac{1}{2}$}{tau greater than 1/2}}  \label{complete} 

We consider the standard model first. 
Note that if $\frac{1}{2}<\tau\leq \frac{w+1}{2w+1}$ then the process is identical to that for $\tau=\frac{1}{2}$. We therefore assume in what follows that $w$ is sufficiently large to ensure $\tau > \frac{w+1}{2w+1}$, which is equivalent to the condition that adjacent nodes of opposite types cannot both be happy. We show that with probability $1-\epsilon$ the starting configuration is such that complete segregation results with probability 1, where $\epsilon \rightarrow 0$ as $n\rightarrow \infty$.  Recall that complete segregation refers to any configuration in which all $\alpha$ nodes belong to a single run, and that, as observed in Section \ref{notation}, once a completely segregated configuration is reached all future configurations must be completely segregated.

 Let $x$ be the number of $\alpha$ nodes in the initial configuration, and let $x_p=x/n$. Our task is to show that for sufficiently large $n$ it is possible to reach a completely segregated configuration from any other, so long as $x_p$ is close to $\frac{1}{2}$. To prove this, however, one must be able to ensure the existence of unhappy individuals of both types at each step along the way. Surprisingly,  the following lemma turns out to be the most tricky aspect of proving Theorem \ref{5}. The proof is given in Section \ref{def3}. The mention of intervals of length $4w+1$ in the statement of the lemma is a just a technicality which will be convenient later.

 \begin{lem} \label{combo} With probability $1-\epsilon$ the initial  configuration on $n$ nodes will have a value $x$ sufficient to ensure  unhappy nodes of both types  outside any interval of length $4w+1$ in \textbf{any} configuration with $x$-many  $\alpha$ nodes and $(n-x)$-many $\beta$ nodes (regardless of ordering),  where $\epsilon \rightarrow 0$ as $n\rightarrow \infty$. 
 \end{lem}

From now on, then,  we assume that all configurations considered have unhappy nodes of both types outside any  given interval of length $4w+1$.  We build a list of configurations from which it is possible to reach complete segregation. Let $\Theta_s(u)$ denote the bias of $u$ at the end of stage $s$. Note that in order for a swap between an unhappy $\alpha$ node $u$ and an unhappy $\beta$ node $v$ to be legal at stage $s+1$,  it suffices either that $v\notin \mathcal{N}(v)$ and $\Theta_s(v)\geq \Theta_s(u)-2$ (since then $v$ changing type will increase the bias at that position by 2 to become at least $\Theta_s(u)$, and similarly $u$ changing type will decrease the bias by 2 to become at most $\Theta_s(v)$), or else that $v\in \mathcal{N}(u)$ and $\Theta_s(v)\geq \Theta_s(u)$. One consequence of this, is that an unhappy $\alpha$ node can always legally swap into the position next to a happy $\alpha$ node if occupied by a $\beta$ node. 

 Consider first any configuration which is not completely segregated, but which has a run of length at least $2w$. Without loss of generality, suppose that this is a run of  $\alpha$ nodes  occupying the interval $[a,b]$, where this interval is chosen to be of maximum possible length. Our aim is to show that from this configuration, one may legally reach another with greater maximal run length.   Now if the nodes $a$ and $b$ are both happy then the length of the interval ensures that all nodes in the run are happy -- this follows by induction on the distance from the edge of the interval by considering the difference between successive neighbourhoods. In this case let $u$ be an unhappy $\alpha$ node (whose existence we have assumed is guaranteed) and let $c\in \{ a,b \}$ be distance at least $w+1$ from $u$.  Then $u$ and the $\beta$ neighbour of $c$ may legally be swapped, increasing the length of the run by at least 1.

 So suppose instead that at least one of the individuals  $a$ and $b$ is not happy, and without loss of generality  suppose that $a$ has bias less than or equal to  $b$. Then $a$ and  $b+1$ may legally be swapped. Performing this swap causes position $b+1$ to have at least the same bias as $b$ did before the swap, and causes  $a+1$ to have at most the same bias as $a$ did before the swap. Thus, the swap has the effect of shifting the run one position to the right and may be repeated until the length of the run is increased by at least 1, i.e.\  for successive $i\geq 0$ we can swap the nodes $a+i$ and $b+i+1$, so long as the latter is of type $\beta$. The first stage at which the latter is of type $\alpha$ the length of the run has been increased.  Putting these observations together,  we conclude that from any configuration which  has a run of length at least $2w$ it is possible to reach full segregation. 
   
   Next consider a configuration in which the longest run $[a,b]$ is of length at least $w$, but strictly less than $2w$. We shall suppose that $[a,b]$ contains $\alpha$ nodes, the case for $\beta$ nodes is similar. Let $c$ be the first $\alpha$ node strictly to the left of $a$. If $c$ is unhappy, then we may legally swap $c$ and $a-1$, strictly increasing the length of the longest run. In order to see this, compare the neighbourhood of $c$ before the swap with the neighbourhood of $a-1$ after the swap. If $|c-a|>w$, then prior to the swap there are at most $w+1$ many $\alpha$ nodes in $\mathcal{N}(c)$, while $\mathcal{N}(a-1)$ has $w+1$ many after the swap. If $|c-a|\leq w$, then in forming $\mathcal{N}(a-1)$ from $\mathcal{N}(c)$, we lose some nodes to the left of $c$, some of which may be $\beta$ nodes, and gain some to the right of $a-1$, all of which are $\alpha$ nodes. If $c$ is happy then the distance between $c$ and $a$ is at most $w$ (otherwise $c$ would only have $\alpha$ nodes in the left half of its neighbourhood) and we may successively swap  unhappy $\alpha$ nodes from outside the interval $[c-w,a +2w]$ with  the nodes $c+i$ for $1\leq i <a-c$ (starting with $i=1$ and proceeding in order), in order to strictly increase the length of the longest run. This follows because with the $i$th swap $c+i$ becomes a happy $\alpha$ node.

    It remains to show that we can always move to a configuration with a run of length at least $w$. 
    Given any configuration (satisfying the condition that any other configuration reached by swapping unhappy nodes has unhappy nodes of both types outside any  given interval of length $4w+1$), we first of all perform a procedure which selects a number of individuals with relatively unfavourable bias. Let $u_0$ be  an $\alpha$ node with least possible bias. Given $u_k$ for $k<w^2$ choose $u_{k+1}$   outside $\bigcup_{j\leq k} \mathcal{N}(u_j)$ which has least possible bias amongst all the $\alpha$ nodes at such sites (we can suppose that $n$ is large enough that such a choice is possible). Once $u_k$ is defined for each $k\leq w^2$, choose $v_0$  outside   $\bigcup_{j\leq w^2} \mathcal{N}(u_j)$, which has the greatest possible bias amongst all $\beta$ nodes at such sites. Given $v_k$ for $k<w^2$ choose $v_{k+1}$ outside  $\bigcup_{j\leq w^2} \mathcal{N}(u_j)$ and outside $\bigcup_{j\leq k} \mathcal{N}(v_j)$ which has greatest possible bias amongst all the $\beta$ nodes at such sites. Now choose an interval $[a,b]$ of length $w$ such that $[a-w, b+w]$  has no intersection with any of the neighbourhoods $\mathcal{N}(u_j)$ or $\mathcal{N}(v_j)$ for $j\leq w^2$ (again assuming $n$ is sufficiently large).

    The point of this procedure was to provide a pool of individuals which we can use for legal swaps.  Now we perform another  iteration, which produces a run in the interval $[a,b]$. At any point during the iteration we say that a neighbourhood $\mathcal{N}(u_j)$ or $\mathcal{N}(v_j)$ is \emph{tarnished} if any node in that neighbourhood has been involved in a swap since the iteration began.

\textbf{ Step 0}. Let $\gamma=\alpha$ if  $a$ is of type $\alpha$, otherwise let $\gamma =\beta$. At any given point in the iteration the value of $\gamma$ specifies the type of run that we are looking to produce, which may change as the iteration progresses.  Set $i:=0$. 

\textbf{Step $s>0$}. We are given a configuration in which all nodes in $[a,a+i]$ are of type $\gamma$, and no nodes in $[a-w, b+w]$ have yet been involved in swaps, except possibly those in $[a,a+i]$. Also, at most $i^2$ of the neighbourhoods $\mathcal{N}(u_k)$ or $\mathcal{N}(v_k)$ are tarnished.  Let $j$ be the least $>i$ such that  $a+j$ is not of type $\gamma$. If $j\geq w$ then the iteration is  complete, and we carry out no further instructions. Otherwise, if  $a+j$ is unhappy, we divide into two subcases. If $a+j-1$ is happy then we can legally swap any unhappy $\gamma$ node from outside the interval $[a-w, b+w]$ with $a+j$. If  $a+j-1$ is unhappy,  we can select any $u_k$ or $v_k$ of type $\gamma$ whose neighbourhood is not tarnished and legally swap that node with  $a+j$. In order to see why this is the case suppose that $\gamma=\alpha$, the case for $\beta$ is similar. Firstly, the fact that at most $i^2$ of these neighbourhoods are tarnished means that we  certainly have untarnished neighbourhoods of the appropriate type to choose from. Also, before the iteration began any $u_k$  had at most bias $\Theta(\mathcal{N}(a+j-1))$ (as that value was defined then). At the present point in the iteration, $\mathcal{N}(u_k)$ being untarnished means that the neighbourhood $\mathcal{N}(u_k)$ is unchanged, while any node in the interval $[a-w, b+w]$ which is now of different type than before, is presently of type $\alpha$.  Swapping $u_k$ with $a+j$ will cause it to have bias at least the same as that which $a+j-1$ had before the swap.  In either of these two subcases, once the swap is performed,  redefine $i:=j$ and proceed to step $s+1$. The final case to consider is that the individual at $a+j$ is happy. In this case, if $\gamma=\alpha$ then redefine $\gamma =\beta$, or if instead $\gamma=\beta$ then redefine $\gamma=\alpha$.  For successive values of $k$, $1\leq k \leq j$,  swap an unhappy $\gamma$ individual from outside the interval $[a-w,b+w]$ with $a+j-k$, causing that node to become happy.  Once this sequence of swaps is complete, redefine $i:=j$ and proceed to step $s+1$.  

\subsection*{The simple model} We wish to show that, whatever the initial configuration, with probability 1 a configuration in which all nodes are of the same type is reached. For $\gamma \in \{ \alpha, \beta \}$, if $\gamma =\beta$ then let $ \gamma^{\ast}=\alpha$, and otherwise let $ \gamma^{\ast}= \beta$.  For the purposes of this discussion, we shall say that $\gamma$ is a \emph{minority type} if there are at most as many nodes of type $\gamma$ as of type $\gamma^{\ast}$. The proof of Lemma \ref{combo} actually suffices to show that if $\gamma$ is a minority type, then there exists at least one $\gamma$ node which is unhappy, and which would have at least as many neighbours of its own (new) type if it changed type. Given any configuration, one can then select $\gamma$ of minority type and successively select nodes of type $\gamma$ which can legally have their type swapped, until all nodes are of type $\gamma^{\ast}$.  

\section{Proofs deferred from Section \ref{taulessthankappa}} \label{def1} 

\begin{replemma}{lem:binom}
Suppose $h: \mathbb{N} \to \mathbb{N}$ and $p \in (0,1)$ are such that there exists $k \in (0,1)$ so that for all large enough $N$, we have $\left( 1+ \left( \frac{1}{p} -1 \right) k \right) h(N)> N \geq h(N)> pN > 0$. Then for all large enough $N$, if $X_N \sim b(N,p)$, we have 
$$P \left( X_N = h(N) \right) \ \ \leq \ \  \textbf{P} \left( X_N \geq h(N) \right) \ \ \leq \ \ \left( \frac{1}{1-k} \right) \cdot \textbf{P} \left( X_N = h(N) \right).$$
That is to say in asymptotic notation, $\textbf{P} \left( X_N \geq h(N) \right)= \Theta \left( \textbf{P} \left( X_N = h(N) \right) \right)$.
\end{replemma}

\begin{proof}
Fix $N$, and for the duration of this proof let $h$ denote $h(N)$. The first inequality is self-evident. For the second, $\textbf{P} \left( X_N \geq h \right) = \ds \sum_{j=h}^N p^j (1-p)^{N-j} \bpm N \\ j \epm $. Looking at the ratio  between successive terms of this sum, we find
$$p^{j+1} (1-p)^{N-j-1} \bpm N \\ j+1 \epm \Big{/} \left( p^{j} (1-p)^{N-j} \bpm N \\ j \epm \right) = \frac{p}{1-p} \cdot \frac{N-j}{j+1}.$$

What is more, this ratio decreases at each step, since $\frac{N-j}{j+1}> \frac{N-(j+1)}{j+2}$. Thus for all $j = h, \ldots, N$, we have $\frac{p}{1-p} \cdot \frac{N-j}{j+1}\leq  \frac{p}{1-p} \cdot \frac{N-h}{h+1} <k$. Thus
$$\textbf{P} \left( X_N \geq h(N) \right) < p^h (1-p)^{N-h} \bpm N \\ h \epm \sum_{j=0}^{N-h} k^j < p^h (1-p)^{N-h} \bpm N \\ h \epm \left( \frac{1}{1-k} \right).$$
\end{proof}

\begin{replemma}{gen}  

Let $P_u$ and $Q_u$ be events which only depend on the neighbourhood of $u$ in the initial configuration, meaning that if the neighbourhood of $v$ in the initial configuration is identical that of  $u$ (i.e. for all $i \in [-w,w]$, $u+i$ is of the same type as $v+i$), then $P_u$ holds iff $P_v$ holds and $Q_u$ holds iff $Q_v$ holds. Suppose also that:
\begin{enumerate}[(i)]  
\item $\textbf{P}(P_u)\neq 0$ and $\textbf{P}(Q_u)\neq 0$. 
\item  For all $k$, for all sufficiently large $w$, $\textbf{P}(P_u)/\textbf{P}(Q_u) >kw$.
\end{enumerate}   
For any $u$, let $x_u$ be the first node to the left\footnote{By the first node to the left of $u$ satisfying a certain condition we mean the first in the sequence $u,u-1,u-2,\cdots$  which satisfies the condition.}  of $u$ such that either $P_{x_u}$ or $Q_{x_u}$ holds.  For any $\epsilon>0$, if $0\ll w \ll n$ then the following occurs with probability $>1-\epsilon$ for $u$ chosen uniformly at random:  $x_u$ is defined and for no node $v$  in $[x_u-2w,x_u]$ does $Q_v$ hold.

An analogous result holds when `left' is replaced by `right'. 
\end{replemma} 
\begin{proof}  For a general node $u$, we'll say it is of \emph{type 1} if $P_u$ holds and no node $u' \in [u - 2w, u]$ satisfies $Q_{u'}$, and of \emph{type 2} if $Q_u$ holds. Let $\pi$ be the probability, for $u$ chosen uniformly at random, that $x_u$ is
defined and of type 1. Our aim is to show that $\pi > 1 -\epsilon$, for all $0 \ll w \ll  n$.
We define an iteration which assigns colours to nodes as follows. 

Step 0. Pick a node $t_0$ uniformly at random. 

Step $s+1$. Let $v_s$ be the first node to the left of $t_s$, such that $v_s=x_{t_s}$ or such that $s>0$ and $v_s=t_0$. Carry out the instructions for the first case below which applies: 

\begin{enumerate} 
\item If there exists no such $v_s$ then terminate the iteration. 
\item If $v_s=t_{0}$ and $s>0$ then make $t_s$ undefined  and terminate the iteration.

\item If there exist any nodes $x$ in $[v_s-2w,v_s]$ such that $Q_x$ holds, then colour $t_s$ black, otherwise it is of type 1 and we colour it white.  Define $t_{s+1}=v_s-(2w+1)$, unless $t_0$ lies in the interval $[v_s-(2w+1),v_s)$, in which case terminate the iteration. 
\end{enumerate} 

This completes the description of the iteration. Let  $S$ be the maximum value of $s$ for which $t_s$ is defined
and coloured. \\

First note that hypothesis (i) in the statement of the lemma guarantees that $S \rightarrow \infty$ as $n \rightarrow \infty$. Similarly, we may assume that at least one $t_s$ is coloured black. Now let $\pi'$ be the proportion of the
$t_s$ which are coloured white, and observe that with high probability,  $\pi'\rightarrow \pi$ as $n\rightarrow \infty$, i.e.\ for all $\epsilon'>0$, for all  $0\ll n$, we have $|\pi-\pi'|<\epsilon'$ with probability $>1-\epsilon'$. In order to see this, consider the situation at the beginning of step $s+1$ of the iteration, when $s>0$. In order to define $t_s$, we moved left until finding  $x_{t_{s-1}}$, and then moved a further $2w+1$ nodes to the left. So far, then, nothing that has happened in the iteration tells us anything about the neighbourhood of $t_s$ and all those nodes to the left of $t_s$ and strictly to the right of $t_0$. So long as $S>s$, the way in which $t_s$ is coloured depends only on these nodes.

So, we now wish to show that for all sufficiently large $w$, with probability tending to $1$ as $n \rightarrow \infty$, we
have $\pi' > 1 - \epsilon$. Let $\rho$ be the ratio of type 1 nodes to type 2 nodes (in $[0,n-1]$), and consider the interval $(t_{s+1}, t_s]$. However $t_s$ is coloured, we have at most $(2w + 1)$ many type 1 nodes in here, namely some subset of
$[v_s-2w, v_s]$. Similarly, if $t_s$ is coloured black then we get at least one type 2 node in here. Thus, summing
over all intervals $(t_{s+1}, t_s]$, we find that
$$\rho \leq \frac{(2w+1)S}{|\{ s \leq S \, : \, t_s \textrm{ is black } \}|} \ \ \mbox{which is to say } \ \  \rho \leq \frac{2w+1}{1 - \pi'}.$$

So if we can show that $\rho \gg  2w + 1$, it will follow that $\pi'$ is close to $1$ as required.
Let $p_1$ and $p_2$ be the probabilities that uniformly randomly selected $u$ itself is of type 1 and 2 respectively. Then, by the weak law of large numbers, we have $\left| \rho - \frac{p_1}{p_2}\right| < \epsilon'$ with probability approaching 1 as $n\rightarrow \infty$. Now property (ii) from the statement of the lemma means that for each $k$, for all sufficiently large $w$, $\frac{p_1}{p_2} > kw$, giving the result.
\end{proof}

\section{Proofs deferred from Section \ref{hardtau}} \label{def2}     

\subsection{Proving Lemma \ref{stableclear}} First, let us restate the result:

\begin{replemma}{stableclear}  
For fixed $k_0$ and $\epsilon'>0$,  $\textbf{P}(\mbox{\textbf{Stable clear}})>1-\epsilon'$ for all  $w$ sufficiently large (and all $n$ sufficiently large compared to $w$). 
\end{replemma} 
\begin{proof} 
For any $u$, let $x_u$ be the first node to the left of $u$ which, in the initial configuration,  either has high bias or else belongs to an interval which is $\tau_0$-stable.  Recall that  $\kappa <\tau_0<\tau$ and that  $k_0$ is fixed while we take $w$ large. Note also that,  at step $i$ of the iteration which defines the sequence $l_1,l_2,..$, the fact that $l_i$ has borderline bias tells us nothing about the neighbourhood of $l_i-(2w+1)$ or the nodes to the left of this neighbourhood (but at a distance small compared to $n$). It therefore suffices to show that for any $\epsilon'>0$, if $0\ll w \ll n$ then the following occurs with probability $>1-\epsilon'$ for $u$ chosen uniformly at random:  $x_u$ is defined and no node in $[x_u-2w,x_u]$ belongs to a $\tau_0$-stable interval (and that an analogous result holds when `left' is replaced by `right'). 

 In order to see this, consider for a moment working with $\tau_0$ rather than $\tau$. In Section \ref{taulessthankappa}  we showed that for any $k$ and for sufficiently large $w$, the probability that a randomly chosen node is unhappy (with respect to $\tau_0$ rather than $\tau$)  is more than $kw$ times the probability that a randomly chosen node belongs to a $\tau_0$-stable interval. Now, since $\tau>\tau_0$, this only makes unhappy nodes \emph{more} likely in the initial configuration. Thus, working with $\tau$, for any $k$ and for all sufficiently large $w$,  the probability that a randomly chosen node is unhappy is more than $kw$ times the probability that a randomly chosen node belongs to a $\tau_0$-stable interval. The result then follows from Lemma \ref{gen}.  
\end{proof}

\subsection{Proving Lemma \ref{lismooth}} We shall need to prove an intermediate lemma first. 
 It turns out that in order to get around the fact that we can't immediately guarantee the required version of condition (ii) of Lemma \ref{gen},  what we need is a bound on the number of nodes of borderline bias which can be expected in the neighbourhood of $u$ such that $\mathtt{Bb}^{\ast}(u)$ holds. The following lemma is a step in this direction:

\begin{defin}
Given $k\geq 1$, let $\mathcal{N}_k(u)$ be the interval $[ u-\lceil w/k \rceil, u+ \lceil w/k \rceil]$. $\mathtt{Bb}^{\ast}(u,k,z)$ holds if $\mathtt{Bb}^{\ast}(u)$ holds and there are at most $z$ many nodes with borderline bias in  $\mathcal{N}_k(u)$ in the initial configuration. 
\end{defin}

\begin{lem} \label{z} Suppose we are given only that $\mathtt{Bb}^{\ast}(u)$ holds. 
For any  $\epsilon'>0$ there exists $z$ such that for $0\ll k \ll w$,  $\mathtt{Bb}^{\ast}(u,k,z)$ holds with probability $>1-\epsilon'$. 
\end{lem}    
\begin{proof} 
We consider the case that $u$ has positive bias $\rho$, and let $\theta$ be the proportion of the nodes in $\mathcal{N}(u)$ which are of type $\alpha$ (so  $\rho=(2 \theta -1)(2w+1)$). The case that $u$ has negative bias is almost identical. 

First of all, we want to show that for sufficiently large $z$, if we step $\lfloor z/2 \rfloor $ many nodes to the left (or right) of $u$, then we will very probably have bias which is well below $\rho$. The argument here is very similar to the proof of Corollary \ref{smooth2} -- in forming the neighbourhood of $v=u-\lfloor z/2 \rfloor $ we lose $\lfloor z/2 \rfloor $ many  nodes from $\mathcal{N}(u)$, with the proportion of $\alpha$ nodes being close to $\theta$, and we gain the same number of new nodes, with the proportion of $\alpha$ nodes here being close to $\frac{1}{2}$. Arguing more precisely,  from amongst the nodes that we lose from $\mathcal{N}(u)$, the expected number of $\alpha$ nodes, $x_0$ say,  is $\theta \lfloor z/2 \rfloor $. By applying  Chebyshev's Inequality just as in the proof of Lemma \ref{smooth}, we conclude that for any $\epsilon''>0$ and for all sufficiently large $z$, $\textbf{P}(|(x_0/\lfloor z/2 \rfloor) -\theta| >\epsilon'') \ll\epsilon'$.  Now consider  $x_1$,  the number of $\alpha$ nodes in $\mathcal{N}(v) \backslash \mathcal{N}(u)$. The weak law of large numbers tells us that for  any $\epsilon''>0$ and for all sufficiently large $z$, $\textbf{P}(|(x_1/\lfloor z/2 \rfloor) -\frac{1}{2}| >\epsilon'' )\ll\epsilon'$. Combining these facts gives that for any $m>0$ and for all sufficiently large $z$:

\[ \textbf{P}(\rho- \Theta(\mathcal{N}(v))<m) \ll \epsilon'.\] 

So far then, we have considered moving $\lfloor z/2 \rfloor $ many nodes to the left of $u$, and have concluded that the bias at this node $v$ will very probably be well below $\rho$ (a similar argument also applies, of course, moving to the right). Now we have to show that as we move left from $v$, so long as we remain within $\mathcal{N}_k(u)$ the bias will very probably remain below $\rho$.  In order to do this, we approximate the bias as we move to the left, by a biased random walk. Recall that for a biased random walk starting at value $-m\leq -1$, with probability $p> \frac{1}{2}$ of going down at each step and probability $1-p$ of going up, the probability of ever hitting 0 is: \[ \left( \frac{1-p}{p} \right)^{m}. \]  
So let us briefly adopt the approximation that  nodes in $\mathcal{N}(u)$ are i.i.d.\ random variables, each of which has probability $\theta$ of being of type $\alpha$. Then, as we move left one position from a location in $\mathcal{N}_k(u)$ to the left of $u$, the probability that the bias increases by 2 is $\frac{1}{2}(1-\theta)$, the probability that the bias remains the same is $\frac{1}{2}$, and the probability that the bias decreases by 2 is $\frac{1}{2}\theta$. Removing those steps at which the bias does not change, we get a biased random walk with probability $\theta$ of going down at each step and probability $1-\theta$ of going up. Choose $\theta'$ with $\frac{1}{2}<\theta' <\theta$. Now, dropping the false assumption of independence, by taking $k$ sufficiently large we ensure that as we take successive steps left from $v$ inside the interval $\mathcal{N}_k(u)$, at each step, no matter what has occurred at previous steps, the probability of the bias increasing is less than $\frac{1}{2}(1-\theta')$, the probability of the bias remaining the same is $\frac{1}{2}$, and the probability of the bias decreasing is greater than $\frac{1}{2}\theta'$. Thus, if the bias at $u-\lfloor z/2 \rfloor$ is $\leq \rho-m$, then the probability that any nodes in the interval $[u-\lceil w/k \rceil, u-\lfloor z/2 \rfloor)$ have high bias is less than $(\frac{1-\theta'}{\theta'})^m$. 

Finally, let $m$ be such that  $(\frac{1-\theta'}{\theta'})^m\ll \epsilon'$, and let $z$ be sufficiently large that, for $v=u-\lfloor z/2 \rfloor$,  
$\textbf{P}(\rho- \Theta(\mathcal{N}(v))<m) \ll \epsilon'$. 
\end{proof}

The reader might observe that the proof of Lemma \ref{z} actually establishes something stronger, and perhaps also more natural, than claimed. Suppose we are given that $\mathtt{Bb}^{\ast}(u)$ holds. 
The proof suffices to show that for any  $\epsilon'>0$ there exists $z$ such that, for $0\ll k \ll w$,  the probability there are more than $z$ nodes of high bias in $\mathcal{N}_k(u)$ is less than $\epsilon'$. One might reasonably wonder why we did not define $\mathtt{Bb}^{\ast}(u,k,z)$ to reflect this stronger condition -- that there are at most $z$ many nodes of \emph{high} bias (rather than borderline bias) in $\mathcal{N}_k(u)$. The reason is that in later counting arguments we shall need failure of the condition (in specific circumstances) to guarantee  that there are, in fact, at least $z$ many nodes of borderline bias in the relevant neighbourhood for which the condition fails.

Of course, Lemma \ref{z}, only restricts the number of nodes of borderline bias that will normally occur in the interval $\mathcal{N}_k(u)$. We have to be able to deal with a larger interval:

\begin{defin} 
We say that $\mathtt{Bb}^{\ast}(u,z)$ holds if $\mathtt{Bb}^{\ast}(u)$ holds and there are at most $z$ many nodes of borderline bias in the interval $[u-(2w+1), u+(2w+1)]$ in the initial configuration. 
\end{defin} 

Now note that if $k_2$ is sufficiently large compared to $k_1$, if $\epsilon'$ is sufficiently small, and if  $\mathtt{Bb}^{\ast}(u,k_1,z)$ and $\mathtt{Smooth}^{\ast}_{k_2,\epsilon'}(u)$  both hold, then in the initial configuration there are at most $z$ many nodes of borderline bias in the interval $[u-(2w+1), u+(2w+1)]$.  Applying Lemmas \ref{z} and \ref{smooth2} we therefore have:  

\begin{coro}[Few nodes of borderline bias] \label{hb}  Suppose we are given that $\mathtt{Bb}^{\ast}(u)$ holds.  For any $\epsilon'>0$, $z$ may be chosen so that for all sufficiently large $w$, $\mathtt{Bb}^{\ast}(u,z)$ holds with probability $>1-\epsilon'$.  
\end{coro}

We are now finally ready to prove Lemma \ref{lismooth}, which we now restate: 

\begin{replemma}{lismooth}
For any node $u$, let $x_u$ be the first node to the left of $u$ which has high bias in the initial configuration.  For any $\epsilon'>0$ and $k\geq 1$, if $0 \ll w \ll n$ and $u$ is chosen uniformly at random, then  $x_u$ is defined and  $\mathtt{Smooth}^{\ast}_{,k,\epsilon'}(x_u)$ holds with probability $>1-\epsilon'$. 

An analogous result holds when `left' is replaced by `right'. 
\end{replemma}
\begin{proof} 

 Applying  Corollary \ref{hb},  choose $z$ such that, for sufficiently large $w$, if we are given that $v$ has borderline bias, then  $\mathtt{Bb}^{\ast}(v,z) $ fails to hold with probability $\ll \epsilon'$, i.e.\ putting $\epsilon_1=\textbf{P}(\neg \mathtt{Bb}^{\ast}(v,z)| \mathtt{Bb}^{\ast}(v))$,  choose $z$ so that $\epsilon_1\ll \epsilon'$ for sufficiently large $w$.  
 Let  $\epsilon_2 =\textbf{P}(\neg \mathtt{Smooth}_{k,\epsilon'}^{\ast}(v)|\mathtt{Bb}^{\ast}(v))$. Then Corollary \ref{smooth2} gives that    for all sufficiently large $w$, $\epsilon_2\ll \epsilon'/z$ -- just apply the statement of the corollary to $k$ and $\epsilon''\ll \epsilon'/z$. 
 
 Now define \[ \mu := \frac{\textbf{P}( \mathtt{Smooth}_{k,\epsilon'}^{\ast}(v)| \mathtt{Bb}^{\ast}(v,z))}{\textbf{P}(\neg \mathtt{Smooth}_{k,\epsilon'}^{\ast}(v) |  \mathtt{Bb}^{\ast}(v,z))}. \] 
  Then 
 \[ \mu = \frac{\textbf{P}( \mathtt{Smooth}_{k,\epsilon'}^{\ast}(v)\wedge \mathtt{Bb}^{\ast}(v,z)   |\mathtt{Bb}^{\ast}(v))}{\textbf{P}(\neg  \mathtt{Smooth}_{k,\epsilon'}^{\ast}(v)\wedge \mathtt{Bb}^{\ast}(v,z)  |\mathtt{Bb}^{\ast}(v)))} \geq \frac{ 1-\epsilon_1- \epsilon_2}{\epsilon_2} \gg \frac{z}{\epsilon'}. \]

 Thus for sufficiently large $w$, with probability approaching 1 as $n\rightarrow \infty$, the ratio amongst the nodes $v$ such that $\mathtt{Bb}^{\ast}(v,z)$ holds,  between the number such that  $\mathtt{Smooth}_{k,\epsilon'}^{\ast}(v)$ holds and the number such that  it does not,  is much greater than $z/\epsilon'$.

 Consider the initial configuration. We define an iteration which assigns colours to nodes as follows. 

Step 0. Pick a node $t_0$ uniformly at random. 

 Step $s+1$. Let $v_s$ be the first node $v$  to the left of $t_s$ such that $v=x_{t_s}$ or such that $s>0$ and $v=t_0$. Carry out the instructions for the first case below which applies ($\neg$ denotes negation): 

\begin{enumerate} 
\item If there exists no such $v$ then terminate the iteration, and declare that it has `ended prematurely'. 
\item If $v_s=t_{0}$ and $s>0$ then make $t_s$ undefined  and terminate the iteration. 
\item If $|v_s-t_s|<2w+1$ then colour $t_s$ pink. 
\item If $\mathtt{Bb}^{\ast}(v_s,z)$ and $\mathtt{Smooth}_{k,\epsilon'}^{\ast}(v_s)$ both occur, then give $t_s$ the colours white and blue. 
\item If  $\mathtt{Bb}^{\ast}(v_s,z) $ and $\neg \mathtt{Smooth}_{k,\epsilon'}^{\ast}(v_s)$ both occur,  then give $t_s$ the colours white and red. 
\item If  $v_s$ satisfies $\neg \mathtt{Bb}^{\ast}(v_s,z) $, then give $t_s$ the colour black. 
\end{enumerate} 

In cases (3)--(6), define $t_{s+1}=v_s-(2w+1)$, unless $t_0$ lies in the interval $[v_s-(2w+1),v_s)$, in which case terminate the iteration. 

This completes the description of the iteration. \\

First note that the probability that the iteration terminates prematurely can be made arbitrarily small by taking $n$ large, and similarly that we may assume there are $t_s$ of all colours. 
Now let $S$ be the greatest $s$ such that $t_s$ is defined when the iteration terminates, and let $\pi$ be the proportion of the $t_s$, $s\leq S$, such that $t_s$ is coloured black. 
 Amongst the nodes $u$ which have borderline bias, let $\rho$ be the ratio between the number for which $\mathtt{Bb}^{\ast}(u,z)$  holds and the number for which it does not -- so that for large $n$, $\rho$ can be expected to be close to $(1-\epsilon_1)/\epsilon_1$.   In order to find an upper bound for $\rho$, first let us find an upper bound for the number of nodes $u$ such that $\mathtt{Bb}^{\ast}(u,z)$ holds.  Let $v_s$ be as defined in the iteration. However we colour $t_s$, there can be at most $z$ many nodes $u$  in $I_s:= (v_s-(2w+1),v_s]$ which satisfy $\mathtt{Bb}^{\ast}(u,z)$, giving an upper bound of $(S+1)z$. Now, let us find a lower bound for the number of nodes with borderline bias for which $\mathtt{Bb}^{\ast}(u,z)$ fails. If $t_s$ is coloured pink, then we are not guaranteed any nodes in $I_s$ of borderline bias for which  $\mathtt{Bb}^{\ast}(u,z)$ fails. If $t_s$ is coloured white, then the same applies. If $t_s$ is coloured black, however, we are guaranteed at least $z$ many nodes of borderline bias in  $I_s$ for which $\mathtt{Bb}^{\ast}(u,z)$ fails. We therefore get a lower bound of $(S+1)\pi z$. Thus: 
\[ \rho \leq \frac{ z(S+1)}{\pi(S+1)z} = \frac{1}{\pi}. \] 

 We chose $z$ previously, so that for all sufficiently large $w$, $\epsilon_1\ll \epsilon'$. Since $\rho$ is close to  $(1-\epsilon_1)/\epsilon_1$ for large $n$, we infer that for all sufficiently large $w$, with probability tending to 1 as $n\rightarrow \infty$, we have  $\rho \gg 1/\epsilon'$, so that $\pi \ll \epsilon'$. For large $n$, the probability  $t_s$ is coloured pink is less than $(2w+1)e^{-w/d}$ (with $d$ as in Definition \ref{d}). So, for sufficiently large $w$, with probability tending to 1 as $n\rightarrow \infty$, the proportion of the $t_s$ which are not coloured white is $\ll \epsilon'$. \\

Now let $\beta$ be the proportion of the $t_s$ which are coloured red. Amongst the nodes such that $\mathtt{Bb}^{\ast}(u,z)$ holds, let $\gamma$ be  the ratio between the number for which  $\mathtt{Smooth}_{k,\epsilon'}^{\ast}(u)$ holds,  and the number for which it does not -- so that $\gamma$ can be expected to be close to $\mu$ (as defined earlier) for large $n$. Again, let us form an upper bound for $\gamma$.  However $t_s$ is coloured, we get at most $z$ many nodes in the interval  $I_s$  for which $\mathtt{Bb}^{\ast}(u,z) \wedge \mathtt{Smooth}_{k,\epsilon'}^{\ast}(u)$ holds. If $t_s$ is coloured red then we get at least one node  in the interval  $I_s$ for which $\mathtt{Bb}^{\ast}(u,z) \wedge \neg \mathtt{Smooth}_{k,\epsilon'}^{\ast}(u)$ holds. Thus: 

\[ \gamma \leq \frac{z(S+1)}{\beta(S+1)}= \frac{z}{\beta}. \]

 We previously observed that, for all sufficiently large $w$, $\mu \gg z/\epsilon'$. Thus, for all sufficiently large $w$, with probability tending to 1 as $n\rightarrow \infty$,  we have that $\gamma  \gg z/\epsilon' $, so that   $\beta \ll \epsilon'$. Since we have already established that, for $0\ll w\ll n$, the probability $t_s$ is not coloured white is $\ll \epsilon'$, we conclude that for $0\ll w\ll n$ the probability $t_s$ is coloured blue is $>1-\epsilon'$.  Now for large $n$, the probability that  $t_s$ is coloured blue is less than the probability,  for $u$ chosen uniformly at random, that $v_u$ is defined and $\mathtt{Bb}^{\ast}(v_u,z)\wedge \mathtt{Smooth}^{\ast}_{k,\epsilon'}(v_u)$ holds, which concludes the proof. 
\end{proof}

\subsection{The proofs of Lemmas \ref{complet} and \ref{final}} 

\begin{replemma}{complet} 
For any $\epsilon'>0$, if $0\ll w \ll n$ then for  all $i\in [1,k_0)$, $l_i$ and $r_i$ will (be defined and will) complete with probability $>1-\epsilon'$. 
\end{replemma}  
\begin{proof} 
We form an upper bound for the probability that $l_i$ will fail to complete (the proof for $r_i$ is essentially identical). Suppose that $\mbox{\textbf{Good spacing}}$ holds, and that $w$ is large.   As before, it is convenient to define $l_0=u_0$. Fix $i<k_0$ and let $I_1=[l_{i+1},l_i]$ and $I_2=[l_i,l_{i-1}]$. Let $k$ be the greatest such that, when $1\leq j \leq k$, $I_1(j:k)$ and $I_2(j:k)$ are of length $\geq w+1$. For $1\leq j \leq \lfloor k/2 \rfloor $ define: 
\[ J_j:= I_1(j:k) \cup I_1(j:k)^- \cup I_2(j:k) \cup I_2(j:k)^-. \] 

Figure \ref{fig:thejjs} below shows, as an example, what $J_2$ looks like: 

\begin{figure}
\epsfig{figure=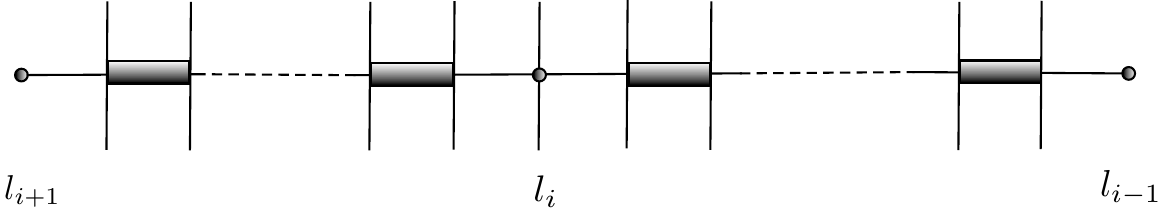,width=11cm}
\caption{$J_2=I_1(2:k) \cup I_1(2:k)^- \cup I_2(2:k) \cup I_2(2:k)^-$.}
\label{fig:thejjs}
\end{figure}

For $1\leq j \leq \lfloor k/2 \rfloor $, let $P_j$ be the event that a node in $J_j$ changes type, and note that $P_3$ cannot occur until $P_2$ has occurred, $P_4$ cannot occur until $P_3$ has occurred, and so on. This follows simply because nodes in $[l_{i+1},l_{i-1}]$ but outside the neighbourhoods $\mathcal{N}(l_{i+1}),  \mathcal{N}(l_{i})$ and  $\mathcal{N}(l_{i-1})$ are initially happy, and cannot become unhappy until another node in their neighbourhood changes type.  Now the basic idea is that if completion fails to occur,  then the sequence of events $P_2,...,P_{\lfloor k/2 \rfloor} $ must occur before any stage at which there are no unhappy nodes in the neighbourhood of $l_i$. 

We label certain stages as being a `step towards completion', and certain others as being a `step towards failure of completion' (while some stages are labelled as neither). These labels do not fully represent all aspects of the process, but contain enough information for our purposes. 

\textbf{Steps towards completion}.  So long as there are unhappy nodes in the neighbourhood of $l_i$, we label any stage at which a node in this neighbourhood swaps type as a  \emph{step towards completion}. Subsequent to any stage at which $l_i$ completes, we label  every stage as a step towards completion.  

 \textbf{Steps towards failure of completion}. If  $1\leq j<\lfloor k/2 \rfloor $ is the greatest  such that $P_j$ has occurred prior to stage $s$ or no $P_j$ has occurred and $j=1$,  and if  $P_{j+1}$ occurs at stage $s$, then we label $s$  a  \emph{step towards failure of completion}.  

Now at any stage $s$ at which some $P_j$ for $ j\leq \lfloor k/2 \rfloor$ is yet to occur, and at which there are unhappy nodes in the neighbourhood of $l_i$, the probability of $s$ being  a step towards failure of completion, is at most $4(w+1)$ times the probability of it being a step towards completion (since there are at most $4(w+1)$ times as many nodes which, if chosen to swap, will cause a step towards failure of completion, as those which will cause a step towards completion). Choosing $d'>d$ we get that, since $\mbox{\textbf{Good spacing}}$  holds, for all sufficiently large $w$, $\lfloor k/2 \rfloor  >e^{w/d'}$. We may therefore consider the first $e^{w/d'}$ many stages which are  steps either towards completion or failure of completion\footnote{In order to ensure the existence of $e^{w/d'}$ many such stages it is momentarily convenient to adopt the convention that the process continues after there are no unhappy nodes $u\in [1,n]$, but with nothing occurring at such stages, and that such stages are also labelled steps towards completion.}  and, for large $w$, consider the  probability that at most $2w$ of these are steps towards completion. By the weak law of large numbers, this probability tends to 0 as $w\rightarrow \infty$.  
Now if there do not exist unhappy nodes of both types in the interval $[l_i-2w, l_i] $ in the initial configuration (which is the  case with probability tending to 1 as $w\rightarrow \infty$, since it holds if $\mathtt{Smooth}^{\ast}(l_i,k',\epsilon'')$ holds for large $k'$ and small $\epsilon''$), then $2w+1$ many steps towards completion prior to $P_{\lfloor k/2 \rfloor}$ occurring, suffices to ensure completion for $l_i$.  \end{proof}

    \begin{replemma}{final} Suppose that all events in $\Pi$ hold.
Let $i$ be the least such that $l_{i}$ is defined and originates a firewall, and let $j$ be least such that $r_{j}$ is defined and originates a firewall. For any $\epsilon'>0$, for all sufficiently large $w$, with probability $>1-\epsilon'$,  $u_0$ will eventually be contained in one of the two firewalls originated at $l_i$ and $r_j$. 
\end{replemma}  
\begin{proof} 
If the firewalls originated at $l_i$ and $r_j$ are of the same type, $\alpha$ say, then the result is immediate -- by Lemma \ref{dichotomy} all $l_{i'}$ for $1\leq i'<i$ and all $r_{j'}$ for $1\leq j'<j$ subside, meaning that no $\beta$-stable intervals can ever be created in $[l_i,r_j]$.   So suppose that the firewall originated at $l_i$ is of type $\alpha$ and the firewall originated at $r_j$ is of type $\beta$. 

As we describe the argument we initially state two facts $(\dagger_1)$ and $(\dagger_2)$ without proof. Once the outline of the argument is complete, we then provide proofs for these facts.  First, note that there are certain type changes within the interval $[l_i,r_j]$ which we are not presently concerned with. If $u$ is in the neighbourhood of some $x$ which is $l_{i'}$ for $1\leq i' <i$ or $r_{j'}$ for $1\leq j'<j$,  and changes type at a stage which is less than or equal to that at which $x$ completes, then we say that this change of type is \emph{previous}. Ignoring changes which are previous, it is then easy to formalise the first stage $s(I)$ at which either of the two firewalls originated at $l_i$ or $r_j$ `have influence' on any given subinterval $I$ of $[l_i,r_j]$: $s(I)$ is the first stage at which any node in $I$ has a change of type which is not previous. 
The fact that all events in $\Pi$ are satisfied means that $s(I)$ must be defined for any subinterval $I$ of  $[l_i,r_j]$.

 For $\gamma \in \{\alpha, \beta \}$ we also define $s^{\gamma}(I)$ to be the first stage at which a node in $I$ has a change of type to $\gamma$ which is not previous -- note that, unlike $s(I)$, these values may be undefined (we write $s^{\gamma}(I)\downarrow$ to indicate that $s^{\gamma}(I)$ is defined).   Now let $u_1=u_0-(2w+1)$ and let $u_2= u_0+(2w+1)$. Let $I_0=\mathcal{N}(u_0)$, let $I_1$ be the left part of $\mathcal{N}(u_1)$, i.e.\ $I_1=[u_1-w,u_1]$  and let $I_2$ be the right part of $\mathcal{N}(u_2)$. 

\begin{enumerate} 
\item[$(\dagger_1)$] For any $\epsilon'>0$, for all sufficiently large $w$, the probability that there does not exist $\gamma \in \{ \alpha, \beta \}$ such that both $s^{\gamma}(I_1)\downarrow = s(I_1)$ and $s^{\gamma}(I_2)\downarrow = s(I_2)$, is $\ll \epsilon'$.
\end{enumerate} 

  So suppose that there does exist such $\gamma$, and suppose $\gamma = \beta$ (the case $\gamma =\alpha$ is similar). 
 
\begin{enumerate} 
\item[$(\dagger_2)$] Let  $v\in [l_i,r_j]$ be the rightmost node such that either $v=l_i$ or  $s^{\alpha}(\mathcal{N}(v))\downarrow <s(I_1)$. For any $\epsilon'>0$, for all sufficiently large $w$, the probability that $|u_0-v|<e^{w/d}$ is $\ll \epsilon'$. 
\end{enumerate} 

  Let us suppose that $|u_0-v|\geq e^{w/d}$. Now consider  any node $u\in I_0$.  Let $I_3$ be the leftmost $w$ many nodes in $\mathcal{N}(u)$, and let $I_4$ be the rightmost  $w$ many nodes. Let $\delta =1-2\tau$, so that if any node $x$ has bias 0 in the initial configuration then $w\delta$ many type changes in $\mathcal{N}(x)$, all to type $\beta$ suffice to give $x$ almost exactly borderline bias (less than borderline bias plus 3 to be more precise).  By the weak law of large numbers, for any $\epsilon'>0$,  for all  sufficiently large $w$, the probability that any node in $[u_1-w, u_2+w]$ has bias in the initial configuration which is of modulus  $>w\frac{1}{3}\delta$ is $\ll \epsilon'$. Suppose that this is not the case.    Now no node in $I_3$ can become unhappy until there have been  $>\frac{2}{3}\delta w$ many changes to $\beta$ type in $I_4$ (with room to spare). Similarly no node in $I_1$ can become unhappy until there have been at least  $\frac{2}{3}\delta  w$ many changes to $\beta$ type in $I_3$. So at stage $s(I_1)$ we conclude that there have been at least $\frac{4}{3}\delta w$ many changes to $\beta$ type in the neighbourhood of $u$. This means that $u$ now has bias at most $w \frac{1}{3} \delta - w\frac{8}{3} \delta$. We conclude that \emph{all} $\alpha$ nodes $u$ in $I_0$ must be unhappy at stage $s(I_1)$. We can now define a notion of completion for $u_0$. We say that $u_0$ completes at stage $s>s(I_1)$ if:

\begin{enumerate} 
\item No node in $I_0$ is unhappy at stage $s$, and this is not true for any $s'<s$ with $s'>s(I_1)$. 
\item  Letting $v$ be defined as in $(\dagger_2)$, there exists $x_0$ with $v+2w<x_0 < u_0 -2w$, such that by the end of stage $s$, no node in $[x_0-w,x_0]$ has had a change of  type which is not previous. 
\end{enumerate} 

 We can then argue, precisely as in Lemma \ref{complet}, that $u_0$ almost certainly completes. In this case it then follows that $u_0$ ultimately belongs to a firewall, which includes the interval $[u_0,r_j]$.  
To complete the proof, we are therefore left to verify $(\dagger_1)$ and $(\dagger_2)$. 

We verify $(\dagger_1)$ in such a way  that almost precisely the same proof suffices to verify $(\dagger_2)$ also. The very basic idea is that, given the large distance between $l_i$ and $r_j$, it is very unlikely that the influence of these two firewalls will first meet sufficiently close to $u_0$ as to cause us trouble.  Given $\epsilon'>0$, choose $k\gg \frac{1}{\epsilon'}$. For all sufficiently large $w$, it follows from our earlier analysis that, in fact, the probability that $|u_0-l_1|<(k+1)\lceil e^{w/d} \rceil $ or $|u_0 - r_1| <(k+1)\lceil e^{w/d}\rceil $ is  $\ll \epsilon'$. So suppose that neither of these possibilities hold. Now, for each instance $\mathcal{Z}$ of the  process on the circle of $n$ nodes  for which there does not exist $\gamma \in \{ \alpha, \beta \}$ such that both $s^{\gamma}(I_1)\downarrow = s(I_1)$ and $s^{\gamma}(I_2)\downarrow = s(I_2)$, we wish to show that there are $k$-many distinct others which we can label as \emph{corresponding} to $\mathcal{Z}$, each of which occurs with the same probability as $\mathcal{Z}$ and for which there \emph{does} exist $\gamma$ of the kind described (and for which all events in $\Pi$  hold for the same $u_0$).  We also require that, for distinct $\mathcal{Z}$ and $\mathcal{Z}'$ for which $\gamma$ as required does not exist, the two corresponding sets of processes have no intersection.   The $k$-many distinct processes corresponding to $\mathcal{Z}$ are easily defined. Given that $\gamma$ does not exist, we must have that $s(I_1)=s^{\alpha}(I_1)$ and that $s(I_2)=s^{\beta}(I_2)$. Note also that (since the `influence' of each of the firewalls originated at $l_i$ and $r_j$ has to spread by one interval of length at most $w$ at a time)  we then have that: 

\begin{enumerate} 
\item[$(\ast)$]  For all $u$ to the left of $I_1$ in $[l_1,r_1]$, $s^{\alpha}(\mathcal{N}(u))\downarrow = s(\mathcal{N}(u))$. Similarly, for all $u$ to the right of $I_2$ in $[l_1,r_1]$, $s^{\beta}(\mathcal{N}(u))\downarrow = s(\mathcal{N}(u))$. 
\end{enumerate} 
The processes corresponding to $\mathcal{Z}$ are the `rotations' of the entire process $\mathcal{Z}$ by $m\lceil e^{w/d}\rceil $ many nodes to the right, for $1\leq m \leq k$.   If $s(I_1)=s^{\alpha}(I_1)$ and $s(I_2)=s^{\beta}(I_2)$, then it follows from $(\ast)$ that this does not hold for the specified rotations. We therefore conclude that the probability that $\gamma$ does not exist as required, is $<\frac{1}{k}\ll \epsilon'$. 

As remarked above, an almost identical proof then suffices to establish $(\dagger_2)$. 
\end{proof} 

\subsection{Dealing with the standard model}  As noted in Section \ref{hardtau}, our aim is to establish the following: 

\begin{enumerate} 
\item[$(\diamond)$] Suppose $w, \tau$ (with $\kappa<\tau <\frac{1}{2})$ and $\epsilon>0$ are fixed (where  $\epsilon$ is the value we have fixed throughout this section, and which played a role in the definition of $k_0$).  Let $l_{k_0}$ and $r_{k_0}$ be as defined previously. For any $\epsilon'>0$, when  $n$ is sufficiently large the following holds with probability $>1-\epsilon'$: there exists a first stage at which there are no unhappy nodes in the interval $[l_{k_0},r_{k_0}]$ and at all stages up to this the total number of unhappy $\alpha$ nodes divided by the total number of unhappy $\beta$ nodes lies in the interval $[1-\epsilon', 1+\epsilon']$.
\end{enumerate}

The very basic idea is as follows. For the remainder of the section we suppose that $w$, $\tau$ and $\epsilon$ are fixed. With probability close to 1, for large $n$  we will have in the initial configuration that the proportion of nodes which are  unhappy and of type $\alpha$ is roughly equal to the proportion of nodes which are unhappy and of type $\beta$.  Briefly, however, let us make the simplification that these proportions   will be \emph{exactly} equal. We then want to show that the process can be sufficiently accurately modelled, for large $n$, by a  system of differential equations, which are entirely symmetric in $\alpha$ and $\beta$. This symmetry means that the solution to the system of differential equations must describe an evolution in which there are always precisely equal numbers of unhappy $\alpha$ and $\beta$ nodes. Of course, we then have to deal with the fact that the numbers  of unhappy  $\alpha$ and $\beta$ nodes in the initial configuration need not actually be exactly equal, but this turns out not to present too many problems.   

As discussed in \cite{BK}, there are, however, some further complications which arise immediately as one looks to apply the Wormald machinery. The method applies to a process in which the state of the system at any given moment in time is described by an $\ell$-dimensional vector of real numbers, where $\ell$ remains fixed, and we then look to approximate the discrete process by a continuous one as $n\rightarrow \infty$.  How can we describe the configuration at any given stage by an $\ell$-dimensional vector, where $\ell$ is independent from $n$? Let $C_n$ be the graph which is a cycle of size $n$.  Up until now, then, we have been considering processes unfolding on each $C_n$. From now on we consider also a value $L$ which depends on $w$, but not on $n$. Generally we shall work under the assumption that $w\ll L\ll n$. For the sake of simplicity we assume that $L$ divides $n$ (but everything that follows is easily modified to deal with the possibility that this is not the case). As we consider the process unfolding on $C_n$, we consider also a parallel process on $G_n$, which is a disjoint union of cycles of length $L$. More precisely, nodes $u$ and $v$ are connected in $G_n$ iff $\lfloor u/L \rfloor = \lfloor v /L \rfloor$ and $u\equiv v \pm 1 \mbox{ mod }  L $.   In order to consider a parallel process on $G_n$, it is also convenient to modify the way in which we count the stages of the process.  In the process as previously described, an unhappy pair of nodes of opposite type are selected at each stage, which may then swap (in fact will swap for the values of $\tau$ considered here). Since we shall now have a situation in which the same node $u$ may be unhappy in $C_n$ but happy in $G_n$, or vice versa, it becomes convenient to consider a process in which two nodes  are selected uniformly at random at each stage, which will only then swap if they are of opposite type and both are unhappy. Of course this makes no real difference to the evolution of the system, except for the way in which we count the stages. The parallel process on $G_n$ then unfolds as follows: when $u$ and $v$ are selected for a potential swap in $C_n$, they swap in $G_n$ if they are both of opposite type and are unhappy in $G_n$. Now let us see how the configuration of $G_n$ at any stage can be described by a $2^L$-dimensional vector. We let $2^L$ denote the set of binary strings of length $L$. For each node $u$ and each stage $s$ we define a string $\tau_{u,s}\in 2^L$: $\tau_{u,s}(0)=1$ if $u$ is of type $\alpha$ at stage $s$, otherwise $\tau_{u,s}(0)=0$, and then $\tau_{u,s}(1)=1$ if the node to the right of $u$ in $G_n$ is of type $\alpha$ at stage $s$, and so on. For each $\sigma\in 2^L$ we define $\zeta_{\sigma}(s)$ to be the number of nodes $u$ such that $\tau_{u,s}=\sigma$. Then the $2^L$-dimensional vector $\boldsymbol{\zeta}(s)$ in which the components are the values $\zeta_{\sigma}(s)$ for $\sigma\in 2^L$ describes the configuration at stage $s$ up to isomorphism.  \\

Our first task now is to show that for large $L$ and large $n$ compared to $L$, the processes on $G_n$ and $C_n$ do not diverge too quickly. To this end we inductively define a set of \emph{tainted} nodes for each stage $s$, denoted $T(s)$. These are  nodes whose neighbourhoods might possibly look different in $G_n$ and $C_n$. $T(0)$ is the set of $u$ such that $-w \leq \ u \mbox{ mod } L < w$ (we assume that the swapping process begins at stage $s=1$). 
If $u$ and $v$ are chosen for a potential swap at stage $s>0$ and are both untainted, then $T(s)=T(s-1)$. Otherwise $T(s)$ is the union of $T(s-1)$ with the set of all nodes which are in the neighbourhood of $u$ or $v$ in either $G_n$ or $C_n$. Immediately it is clear that $|T(s)-T(s-1)|<4(2w+1)$. In fact, by counting more carefully those nodes which belong to the neighbourhoods of $u$ and $v$ in \emph{both} graphs, we get 
$|T(s)-T|(s-1)|<2(3w+1)$, so that, since $w$ is large: 
\begin{equation} \label{firsttaintbound}  
|T(s)-T(s-1)|<7w.
\end{equation} 
 The next lemma provides precisely the kind of probabilistic bound on the number of tainted nodes at any given stage which we will need later. 

\begin{lem} \label{taint bound} The following conditions hold at every stage $s\geq 0$: 
\begin{enumerate} 
\item The expected number of tainted nodes at stage $s$ is bounded above by $e^{14ws/n} \left( \frac{2w-1}{L}\right)n$. 
\item The probability that $\frac{|T(s)|}{n}>2e^{14ws/n} \left( \frac{2w-1}{L}\right)$ is at most $e^{-wn/(21L^2)}$.  
\end{enumerate} 
\end{lem} 
\begin{proof} Let $\rho(s)=|T(s)|$. In order to prove (1) first, let $u$ and $v$ be the nodes chosen for a swap at stage $s+1$. Let $\gamma$ be the probability that either of these nodes are tainted (at the end of stage $s$). The probability that $u$ is tainted is $\rho(s)/n$, and similarly for $v$, so $\gamma \leq 2\rho(s)/n$. Now if neither of $u$ or $v$ is tainted then $\rho(s+1)=\rho(s)$, and otherwise, by (\ref{firsttaintbound}), $\rho(s+1) <\rho(s)+7w$. We therefore have: 
\[ \textbf{E}[\rho(s+1)\ |\ \rho(s)] < (1-\gamma)\rho(s) +\gamma(\rho(s)+7w)  \leq \rho(s) + 2\frac{\rho(s)}{n} 7w = \rho(s)\cdot \left(1+ \frac{14w}{n} \right).\] 
 For $x>0$, $1+x<e^x$, giving: 
\[ \textbf{E}[\rho(s+1)\ |\ \rho(s)] <\rho(s)\cdot e^{14w/n}.\] 
 So far, then, we have established that the sequence of random variables $Y_s=\rho(s)e^{-14ws/n}$ is a supermartingale. This suffices to give (1) since $Y_0 = \left(\frac{2w-1}{L} \right) n$. 

In order to establish (2) we make use of the following version of Azuma's inequality:  If $X_1,...,X_s$ is a supermartingale with $X_0=0$ and $|X_i-X_{i-1}|\leq c_i$ for all $i$ and constants $c_i$, then for all $\Omega>0$, 
\begin{equation} \label{Wormlem} 
 \textbf{P}(X_s\geq \Omega)\leq \mbox{ exp}\left( - \frac{\Omega^2}{2\sum_{i=1}^s c_i^2} \right). 
 \end{equation} 
 If we put $X_s=\frac{Y_s}{n}-\frac{2w-1}{L}$ then the sequence $X_0,X_1,...$ is a supermartingle and $X_0=0$. In order to apply (\ref{Wormlem}) we need a bound on the values $|X_i-X_{i-1}|$.  To this end we consider the sum: 
\[ | Y_s - e^{-14w/n}Y_s| + | Y_{s+1} - e^{-14w/n}Y_s |.\] 
 Now  
\[ Y_s- e^{-14w/n}Y_s = Y_s( 1 -  e^{-14w/n}) < Y_s 14w/n,\] 
 since for $x>0$, $1-e^{-x}<x$.  Also,
\[ Y_{s+1} - e^{-14w/n}Y_s = e^{-14(s+1)w/n} (\rho(s+1)-\rho(s)) \leq 7w  e^{-14(s+1)w/n},\] 
  by (\ref{firsttaintbound}). Thus, 
\[ | Y_s - e^{-14w/n}Y_s| + | Y_{s+1} - e^{-14w/n}Y_s | < 14wY_s/n +7w e^{-14w(s+1)/n} < 21 we^{-14ws/n}. \] 
 We may therefore put $c_s= \frac{21}{n} we^{-14ws/n}$. Now since for $0<x<1$, $\frac{1}{1-x}=1+x+x^2+x^3 \cdots $ we conclude that $\sum_{s=0}^{\infty} c_s^2 = (21w/n)^2(1-e^{-28w/n})^{-1} <21w/n$, the latter inequality following from the fact that for small positive $x$, $1-e^{-x}>3x/4$. 

 Finally, applying (\ref{Wormlem}), we get:  
 
\begin{eqnarray*}
\textbf{P}\left( \frac{\rho(s)}{n} > 2 e^{14ws/n} \left( \frac{2w-1}{L} \right) \right) & = & \textbf{P}\left( \frac{Y_s  e^{14ws/n} }{n} > 2 e^{14ws/n} \left( \frac{2w-1}{L} \right) \right) \\
 &=& \textbf{P} \left( X_s+\left( \frac{2w-1}{L} \right)   > 2 \left( \frac{2w-1}{L} \right) \right) \\
 &= & \textbf{P} \left( X_s > \left( \frac{2w-1}{L} \right) \right) \\
&\leq & \mbox{exp} \left( -\frac{n(2w-1)^2}{42wL^2} \right) \\
&<&  \mbox{exp} \left( - \frac{nw}{21L^2} \right). \\
 \end{eqnarray*}
  This establishes (2) as required. 
\end{proof} 

So Lemma \ref{taint bound} tells us that for fixed $x$, we can ensure with probability close to 1  that the proportion of nodes which are tainted by stage $xn$ is very small, so long as we take $L$ large and $n$ large compared to $L$. With this in place, we now concentrate for a while on the processes on the graphs $G_n$. In order to approximate these processes by the solution to a system of differential equations, we first of all have to draw up reasonable candidates for the differential equations to be used. To this end we begin by considering the conditional expectations $\textbf{E}[\zeta_{\sigma}(s+1)\ | \ \boldsymbol{\zeta}(s)]$. One can form an expression for this expectation, simply by listing all of the different ways in which the number of nodes $u$ with $\tau_u=\sigma$ can increase or decrease at a given stage and establishing their probabilities (we formally defined $\tau_{u,s}$ rather than $\tau_u$, but shall omit $s$ where no ambiguity results). Consider a proposed swap  between $u$ and $v$ which belong to different cycles, such that $\tau_{u}=\sigma'$ and $\tau_{v}=\sigma''$.  Define $a(\sigma,\sigma',\sigma'')$ to be the net change resulting from this proposed swap, in the number of nodes $u'$  for which $\tau_{u'}=\sigma$. Note that the net change will be zero unless $u$ and $v$ are of opposite type and are both unhappy, and that in any case $| a(\sigma,\sigma',\sigma'')| <2L$. Then: 

\begin{equation} \label{sys} 
\textbf{E}[\zeta_{\sigma}(s+1)-\zeta_{\sigma}(s)\ |\ \boldsymbol{\zeta}(s)] = \sum_{\sigma',\sigma''} a(\sigma,\sigma',\sigma'') \frac{\zeta_{\sigma'}(s)\zeta_{\sigma''}(s)}{n^2} \ + \ O\left( \frac{L^2}{n} \right). 
\end{equation} 

 The $O\left( \frac{L^2}{n} \right)$ correction term accounts for the possibility that the nodes selected for a potential swap might belong to the same cycle, the probability of this occurring being $L/n$ and the net change then being bounded by $L$. We next perform some scaling, so that the variables reach fixed functions in the limit as $n\rightarrow \infty$: for each $\sigma$ we want a real function  $z_{\sigma}(s)$ to model the behaviour of $\frac{1}{n}\zeta_{\sigma}(sn)$.  Now (\ref{sys}) \emph{suggests} the following system of differential equations for the functions $z_{\sigma}$: 
\begin{equation} \label{sys2} 
z_{\sigma}'(s)= \sum_{\sigma',\sigma''} a(\sigma,\sigma',\sigma'') z_{\sigma'}(s)z_{\sigma''}(s). 
\end{equation} 
\begin{equation} \label{init} 
 z_{\sigma}(0)=\frac{1}{n} \zeta_{\sigma}(0). 
\end{equation}   

Note that for any solution to (\ref{sys2}) and (\ref{init}), if we put $y(s)=\sum_{\sigma} z_{\sigma}(s)$, then: 

 \[ y'(s)=\sum_{\sigma} z_{\sigma}'(s)= \sum_{\sigma',\sigma''} \sum_{\sigma} a(\sigma,\sigma',\sigma'') z_{\sigma'}(s)z_{\sigma''}(s)= \sum_{\sigma',\sigma''} 0=0. \] 
 
   Thus, since  $y(0)=1$, we also have $y(s)=1$ for all $s$ in the domain. This in turn means that : 
 
 \begin{equation} \label{derivativebound} |z_{\sigma}'(s)| = |\sum_{\sigma'} \sum_{\sigma''} a(\sigma,\sigma',\sigma'') z_{\sigma'}(s) z_{\sigma''}(s) | \leq   |\sum_{\sigma'} \sum_{\sigma''} 2L z_{\sigma'}(s) z_{\sigma''}(s) | = |\sum_{\sigma'} 2L z_{\sigma'}(s)| = 2L 
 \end{equation} 

Suppose for now that we are interested in the first $xn$ steps of the process on $G_n$. Later we shall choose a value for $x$ which suits our needs. For now though, this means that we are interested in solutions to (\ref{sys2}) and (\ref{init}) for $s\in [0,x]$. Let $\Pi$ be the set of points  $(s,z_1,...,z_{2^L})$ such that   $s\in [0,x]$ and, for each $i\leq 2^L$, $|z_i|\leq 2Lx+1$. Let $D$ be a bounded connected open subset of $\mathbb{R}^{2^L+1}$ containing $\Pi$ and such that the minimum $\ell^{\infty}$ distance between the boundary of $D$ and any point in $\Pi$ is bounded below (and where the $\ell^{\infty}$ distance between two vectors is the maximum difference between corresponding components). For the remainder of this section, let all the strings in $2^L$ be enumerated as $\sigma_1,\sigma_2,...,\sigma_{2^L}$. Then by standard results in the theory of differential equations (see \cite{WH} for example), the fact that each of the functions 
\[ f_{i}(s,z_1,...,z_{2^L}):= \sum_{j,k}  a(\sigma_i,\sigma_j,\sigma_k) z_jz_k\] 

  satisfies a Lipschitz condition for each argument $z_r$ in $D$, together with (\ref{derivativebound}), means that  there exists a unique set of functions $\{ z_{\sigma_i}(s) \ |\ \sigma_i\in 2^L \}$ defined on the interval $[0,x]$  and which satisfy (\ref{sys2}) and (\ref{init}) for all $s\in [0,x]$.  

The following lemma is just a slightly stripped down version of Theorem 5.1 from \cite{NW}.

\begin{lem} \label{Worm}  Let $D$ and $x$ be as above.  Suppose that for all $n$, $\lambda(n)>L^2/n$ and that as $n\rightarrow \infty $, $\lambda(n)\rightarrow 0$. Consider the unique set of functions $\{ z_{\sigma}\ |\ \sigma\in 2^L \}$ defined on the interval $[0,x]$ which satisfy (\ref{init}) and (\ref{sys2}) for all $s\in [0,x]$. Then the following holds with probability $1-O(\frac{2^{L+1}L}{\lambda(n)} \mbox{\emph{exp}} ( - \frac{n\lambda(n)^3}{8L^3} ))$: 

\[ \zeta_{\sigma}(s)=nz_{\sigma}(s/n) + O(n\lambda(n)) \] 

 uniformly for $0\leq s \leq xn$.   
\end{lem}

 Actually, in the statement Theorem 5.1 in \cite{NW}, the conclusion is only guaranteed for those  $s$ for which
the vector $(s,z_{\sigma_1}(s),\dots,z_{\sigma_{2^L}}(s))$ is within $\ell^{\infty}$ distance $C\lambda(n)$  of the boundary of $D$ for some sufficiently large constant $C$ (where $C$  is independent from $n$). In our case, however, (\ref{derivativebound}) together with the way in which we defined $D$ removes these complications.  Since we required that the minimum $\ell^{\infty}$ distance between $D$ and $\Pi$ be bounded below, and since $\lambda(n)\rightarrow 0 $ as $n\rightarrow \infty$, the required condition is automatically satisfied for all sufficiently large $n$ and all $s\in [0,x]$ (and then the constant in the $O(\frac{2^{L+1}L}{\lambda(n)} \mbox{exp} ( - \frac{n\lambda(n)^3}{8L^3} ))$ term can be modified to ensure the theorem holds for all $n$). \\

Consider now the symmetry properties which must be satisfied by the functions $z_{\sigma}$. For each $\sigma\in 2^L$ let $\bar \sigma$ be the string which results from changing every 1 in $\sigma$ to 0, and changing every 0 to 1. For $\boldsymbol{z}\in \mathbb{R}^{2^L}$, let $\iota(\boldsymbol{z})$ be the vector which results from swapping the $i$th component with the $j$th component whenever $\bar \sigma_i=\sigma_j$ (recalling our enumeration of the strings of length $2^L$ from before). If we have $\iota(\boldsymbol{\zeta}(0))=\boldsymbol{\zeta}(0)$ (which we probably will not, but for large $n$ we shall have something close to it with good probability), then the symmetry properties of (\ref{sys2}) guarantee that for all $s\in [0,x]$, $\iota(\boldsymbol{\zeta}(s))=\boldsymbol{\zeta}(s)$.
Define $u(\sigma)=1$ if any node $u$ such that $\tau_u=\sigma$ is an unhappy $\alpha$ node, $u(\sigma)=-1$ if any node $u$ such that $\tau_u=\sigma$ is an unhappy $\beta$ node, and $u(\sigma)=0$ otherwise. For $\boldsymbol{z} = (z_1,...,z_{2^L}) \in \mathbb{R}^{2^L}$ consider the linear function: 
\[ \Delta(\boldsymbol{z})= \sum_i u(\sigma_i) z_i. \] 

 Then $\Delta(\boldsymbol{\zeta}(s)/n)$ is the difference at stage $s$, between the fraction of nodes which are unhappy and of type $\alpha$, and the fraction which are unhappy and of type $\beta$.  Let  $\boldsymbol{z}(s)$ be the vector solution to (\ref{sys2}) and (\ref{init}), i.e.\ with $i$th component $z_{\sigma_i}(s)$.  If $\iota(\boldsymbol{\zeta}(0))=\boldsymbol{\zeta}(0)$ then we have that $\Delta(\boldsymbol{z}(s))=0$ for all $s\in [0,x]$.  \\

In the above we considered the case that $\iota (\boldsymbol{\zeta}(0))=\boldsymbol{\zeta}(0)$. We now have to deal with the fact that this will probably not hold exactly.  Recall that we are interested in the processes on $C_n$ and $G_n$ up stage $xn$ for some fixed $x$, which we are yet to specify, but which will not depend on $n$ or $L$. For now suppose that we want to ensure the following modified version of our ultimate aim $(\diamond)$ (one of the differences being that this modified statement refers to a difference in numbers rather than a ratio): 

\begin{enumerate} 
\item[$(\dagger)$] For fixed $x>0$ and any $\epsilon_1>0$ the following holds for the process on $C_n$, with probability  $>1-\epsilon_1$ for all sufficiently large  $n$: at all stages $\leq xn$ the difference between the fraction of nodes which are unhappy and of type $\alpha$ and the fraction which are unhappy and of type $\beta$, is less than $\epsilon_1$.
\end{enumerate}

 To establish $(\dagger)$ take $\epsilon_2<\epsilon_1$.  It suffices to ensure that with probability $>1-\epsilon_2$,  $|\Delta(\boldsymbol{z}(s))| <\epsilon_2$ for all $s\in [0,x]$, since then we can take $L$ large and $n$ sufficiently large compared to $L$ and apply Lemmas \ref{taint bound} and \ref{Worm} (putting $\lambda(n) = \mbox{ max} \{ \frac{L^2+1}{n}, n^{-1/4} \}$, for example) to get the required result for $C_n$. 

Let $S$  denote the set of vectors $\boldsymbol{y}=(y_1,...,y_{2^L})$ such that the solution to (\ref{sys2}) with   $\boldsymbol{z}(0)=\boldsymbol{y}$,  satisfies $|\Delta(\boldsymbol{z}(s))|<\epsilon_2$ for all $s\in [0,x]$. Let $\Sigma =\{ \boldsymbol{y}\ |\ \iota(\boldsymbol{y})=\boldsymbol{y}, \forall i, y_i \geq 0, \sum_i y_i=1 \}$. By the continuity properties of ordinary differential equations (see for example \cite{BN}) it follows that $S$ is an open set, so since $\Sigma$ is compact  there exists $d_0$ such that  every point within distance $d_0$ of $\Sigma$ belongs to $S$. Then it follows directly from the law of large numbers that when $n$ is sufficiently large compared to $L$, $\boldsymbol{\zeta}(0)/n$ is within distance $d_0$ of $\Sigma$ with probability $>1-\epsilon_2$. This establishes $(\dagger)$ as required. \\

Finally, we have to specify the value $x$. Recall that our aim is to prove $(\diamond)$, as specified previously. 
Let $l_1$ be the length of the interval $ [l_{k_0},r_{k_0}]$. Note that (for fixed $w$ and $\epsilon$, where $\epsilon$ is the value for which we are proving Theorem \ref{3}, and which plays a role in the value of $k_0$, and for $\epsilon'$ as in $(\diamond)$),   we can take $ l_2 $ such that for all sufficiently large $n$, the probability that either $|l_{k_0}-u_0|>l_2/2$ or  $|u_0-r_{k_0}|>l_2/2$ is $\ll \epsilon'$. Taking  $l_3 \gg l_2/\epsilon'$,  consider now the first stage $s_1$ at which, for some $ \gamma \in \{ \alpha, \beta \}$, the number of unhappy $\gamma$ nodes as a proportion of  $n$ is less than $1/l_3$. 
Taking $\epsilon''\ll \epsilon'$ and putting  $x> (2w+1)(l_3)^2/\epsilon''$, suppose towards a contradiction that  $\textbf{P}(s_1>xn)>\epsilon''$. In that case, with probability $>\epsilon''$, the following holds at each stage $1\leq s\leq xn$:   given the configuration at the end of stage $s-1$, the probability at stage $s$ of choosing two nodes which do swap  is at least $\frac{1}{(l_3)^2}$. The expected number of stages $\leq xn$ at which swaps do occur is therefore $>(2w+1)n$, which gives the required contradiction since, according to the observation made in Section \ref{notation},  there can be at most $(2w+1)n$ stages at which a swap occurs.  Applying $(\dagger)$ to $\epsilon_1 \ll \epsilon'/l_3$  then establishes that with probability approaching 1 as $n\rightarrow \infty$, $s_1<xn$ and at all stages $\leq xn$ the difference  between the fraction of nodes which are unhappy and of type $\alpha$ and the fraction which are unhappy and of type $\beta$, is less than $\epsilon_1$. When the latter conditions hold, this means that at all stages up to $s_1$ the ratio between the number of unhappy $\alpha$ nodes and the number of unhappy $\beta$ nodes lies in the interval $[1-\epsilon',1+\epsilon']$. Now choosing $l_3\gg l_2/\epsilon'$ as we did previously, means that, since $u_0$ was chosen uniformly at random, and given that the total proportion of nodes which are unhappy at stage $s_1$ is at most $\frac{2+\epsilon'}{l_3}$, the probability that there are any unhappy nodes in the interval $[u_0-\lceil l_2/2\rceil, u_0+\lceil l_2/2 \rceil] $ (or in the almost certainly smaller interval $[l_{k_0},r_{k_0}] $)  at stage $s_1$ is $\ll \epsilon'$. Thus we have established $(\diamond)$, as required.

\section{Proofs deferred from Section \ref{complete}} \label{def3}  

 We must prove Lemma \ref{combo}. Initially, however, we ignore the required conditions concerning intervals of length $4w+1$, and we simply look to establish that when $x_p$ is close to $\frac{1}{2}$ there are unhappy $\alpha$ nodes in any configuration.  Rather than considering the nodes to be arranged in a circle, it is temporarily useful to suppose instead that they are arranged in a line which extends infinitely far to the right, with positions indexed by natural numbers (starting with 0). So we consider a fixed configuration in which each natural number is either of type $\alpha$ or $\beta$. It is convenient to suppose further that the node at position $w$ is of type $\alpha$, and that all  $\alpha$ nodes   $\geq w$ are happy.      From these assumptions, we will deduce the following bound on the proportion of $\alpha$ nodes: 
 
 \begin{enumerate} 
 \item[$(\dagger)$] For any $\epsilon>0 $ there exists $\ell$ such that for all $\ell'\geq \ell$, the proportion of $\alpha$ nodes in the interval $[0,\ell']$ is $\geq \frac{w+1}{2w}-\epsilon$.  
  \end{enumerate} 
  
   This suffices because the existence of any circle in which all $\alpha$ nodes are happy and in which $0<x_p<\frac{w+1}{2w}$ (the latter condition on $x_p$ being satisfied almost certainly for large enough $n$, by the weak law of large numbers) would give a contradiction, since this circle could then be cut and infinitely many copies placed end to end, giving a counter example to $(\dagger)$.
 
In order to prove $(\dagger)$,  it will be useful to consider the right and left parts of any neighbourhood  $\mathcal{N}(u)$ for $u\geq w$. We define $\mathcal{L}(u)$ (respectively $\mathcal{R}(u)$) to be the leftmost (rightmost) $w$-many  nodes in $\mathcal{N}(u)$.  Note  that any $\alpha$ node  $u\geq w$ must have $\alpha$ nodes  in both $\mathcal{L}(u)$ and $\mathcal{R}(u)$. Also, any $\beta$ nodes $v>w$ must  have  $\alpha$ nodes  in $\mathcal{L}(v)$ and in  $\mathcal{R}(v)$, otherwise the next  $\alpha$ node to the left or right (respectively) would not be happy. 
 
 We define a sequence of $\alpha$ nodes. Let $u_0=w$. Given $u_k$ define $u_{k+1}$ to be the rightmost $\alpha$ node in $\mathcal{N}_k:=\mathcal{N}(u_k)$. 
 Also define  $\mathcal{L}_k:=\mathcal{L}(u_k)$ and $\mathcal{R}_k=\mathcal{R}(u_k)$. 
Now, for $m\geq 1$ define $I_m:= \bigcup_{k=0}^{m} \mathcal{N}_k$ and $S_m:= \Sigma_{k=0}^{m} \Theta(\mathcal{N}_k)$. Since each $u_k$ is happy,  $\Theta(\mathcal{N}_k) \geq 3$, giving  $S_m \geq 3(m+1)$. This doesn't immediately tell us the bias on the interval $I_m$, however, because nodes may have been counted multiple times in forming the sum $S_m$. 
We therefore want to consider the way in which the neighbourhoods $\mathcal{N}_k$  overlap. 
For $k \geq 1$, define $\mathcal{N}'_k:=\mathcal{R}_{k-1} \cap \mathcal{L}_{k+1}$ $(=\mathcal{N}_{k-1} \cap \mathcal{N}_{k+1}$).  Notice that $u_k \in \mathcal{N}'_k$, but that $u_{k-1} \not\in \mathcal{N}'_k$ and  $u_{k+1} \not\in \mathcal{N}'_k$. We similarly partition each $\mathcal{N}'_k$ into a left and right part. Define $\mathcal{L}'_k:=\mathcal{N}'_k \cap \mathcal{L}_{k}$ and $\mathcal{R}'_k:=\mathcal{N}'_k \cap \mathcal{R}_{k}$. It is immediate that if $\mathcal{R}'_k$ is non-empty then all nodes in it are of type $\beta$.  
Now we partition the nodes in $I_m$ according to the number of times they are counted in forming the sum $S_m$, i.e.\ the number of $\mathcal{N}_k$ that they belong to for $k\leq m$.  We define  $J^m_r$ to be the set of $u\in I_m$ which belong to exactly $r$ distinct neighbourhoods $\mathcal{N}_k$ for $k\leq m$. Notice that $\bigcup_{r \geq 3} J_r^m \subseteq \bigcup_{k=1}^{m-1} \mathcal{N}'_k$. On the other hand, $u_k \in \bigcup_{r \geq 3}  J_r^m$ for $1\leq k \leq m-1$.  We now want to assess the size and bias of each of the various $J_r^m$. 

$\boldsymbol{J_1^m}$. We have $J_1^m \subseteq \mathcal{L}(u_0) \cup \mathcal{R}(u_m)$, so $|J_1^m| \leq 2w$.  

$\boldsymbol{J_4^m}$. We show that only $\beta$ nodes can belong to $J_4^m$. 
Suppose that $u \in J_4^m$ is of type $\alpha$. Then $u \in \mathcal{N}'_k$ for some $k$ with $1\leq k\leq m-1$, and by the remark above $u \not\in \mathcal{R}'_k$. So $u=u_k$ or $u \in \mathcal{L}'_k$. In either case, $u \in \mathcal{N}_{k-2}$ will contradict the definition of $u_{k-1}$, while $u \in \mathcal{N}_{k+2}$ will contradict the definition of $u_{k+1}$, giving the required contradiction. 

$\boldsymbol{J_r^m}, \ r\geq 5$. Suppose  $u \in \mathcal{N}_{k-2} \cap \mathcal{N}_{k-1} \cap \mathcal{N}_{k} \cap \mathcal{N}_{k+1} \cap \mathcal{N}_{k+2}$. Then either $u_k \in \mathcal{N}_{k-2}$ contradicting the definition of $u_{k-1}$ or $u_{k} \in \mathcal{N}_{k+2}$ contradicting that of $u_{k+1}$.  Thus $J_r^m$ is empty for $r\geq 5$.

$\boldsymbol{J_3^m}$.  We claim that $\Theta(J_3^m)<m+w$.
From our analysis of $J_r^m$ for $r>3$ it follows that we must have $u_k \in J_3^m$ for $1\leq k\leq m-1$. So $\Theta(J_3^m)\leq m - 1 + \Sigma_{k=1}^{m-1} |\mathcal{L}'_k| - \Sigma_{k=1}^{m-1}|\mathcal{R}'_k|$. It therefore suffices to show that $\Sigma_{k=1}^{m-1} (|\mathcal{L}'_k| - |\mathcal{R}'_k|)$ is bounded by $w$. For $k\geq 1$ define  $d_k:=u_{k}-u_{k-1}$. Then $|\mathcal{R}'_k|=w-d_k$ and $|\mathcal{L}'_k|=w-d_{k+1}$. Hence $|\mathcal{L}'_k|=|\mathcal{R}'_{k+1}|$, and $\Sigma_{k=1}^{m-1} (|\mathcal{L}'_k| - |\mathcal{R}'_k|)=|\mathcal{L}'_{m-1}|-|\mathcal{R}'_1| \leq w$, proving the claim.

Putting these facts together, we get that $2\Theta(I_m)> 3m -2w -m-w=2m-3w$, so $\Theta(I_m)> m-\frac{3}{2}w$. For $m>\frac{3}{2}w$ choose a partition, $I_m= \Pi_0 \cup \Pi_1 \cup \Pi_2$, so that $|\Pi_0|=|\Pi_1|$, $\Pi_0$ contains only $\beta$ nodes, and $\Pi_1$ and $\Pi_2$ contain only $\alpha$ nodes (meaning that $|\Pi_2|>m-\frac{3}{2}w$). Then since $|I_m|\leq 2w+1 +mw$, the proportion of the nodes in $I_m$ which belong to $\Pi_2$ is greater than 
\[ \frac{m-\frac{3}{2}w}{2w+1+mw} > \frac{m}{m+3} \cdot \frac{1}{w} - \frac{3}{2m+6} \] 

 which tends to $\frac{1}{w}$ as $m\rightarrow \infty$.  Thus, for each $\epsilon >0$, there exists $m$ such that for all $m'\geq m$, the proportion of $\alpha$ nodes in $I_{m'}$ is $\geq \frac{w+1}{2w}-\epsilon$, giving $(\dagger)$, as required.

 Finally, we have to deal with the fact that the statement of the lemma actually requires the existence of unhappy nodes outside any interval of length $4w+1$. We can get this at the expense of assuming $n$ to be reasonably large compared to $w$. Given any circle configuration  in which all $\alpha$ nodes outside an interval of length $4w+1$ are happy, cut the circle at the left end of the interval of length $4w+1$ and consider the individuals in the circle  to lie in a line with sites indexed by natural numbers $<n$ (so those in the interval of length $4w+1$ occupy $[0,4w]$).  Performing an almost identical analysis  we still conclude that, for  any $\epsilon>0$, so long as $n$ is sufficiently large,  there exists $\ell$ such that for all $\ell'\geq \ell$ (with $\ell'<n$), the proportion of $\alpha$ nodes in the interval $[0,\ell']$ is $\geq \frac{w+1}{2w}-\epsilon$.  We conclude that with probability $1-\epsilon$ the initial  configuration on a circle of size $n$ will have a value $x$ sufficient to ensure  unhappy nodes of both types  outside any interval of length $4w+1$ in \emph{any} configuration with the same number of $\alpha$ and $\beta$ nodes,  where $\epsilon \rightarrow 0$ as $n\rightarrow \infty$.

\end{document}